\documentclass[journal,10pt]{IEEEtran}

\ifCLASSINFOpdf

\else

\fi

\usepackage[utf8]{inputenc}
\usepackage{url}
\usepackage{cite}
\usepackage{hyperref}
\usepackage{graphicx}
\usepackage[nodisplayskipstretch]{setspace}
\usepackage{amsfonts} 
\usepackage{lipsum}
\usepackage{amsmath}
\usepackage{dsfont}
\usepackage{amssymb}
\usepackage{bbm}
\usepackage{amsmath}

\DeclareMathOperator*{\argmin}{arg\,min}
\DeclareMathOperator*{\limisup}{lim\,sup}
\DeclareMathOperator*{\amin}{\mbox{minimize}}
\usepackage[dvipsnames]{xcolor}
\usepackage{csquotes}
\usepackage[utf8]{inputenc}
\usepackage[english]{babel}
\usepackage{dsfont}
\usepackage{amsthm,amssymb}
\newtheorem{theorem}{Theorem}
\newtheorem{corollary}{Corollary}[theorem]

\newtheorem{propsn}{Proposition}
\usepackage{bbm}

\usepackage{algorithmicx}
\usepackage{algpseudocode}
\usepackage[norelsize, linesnumbered, ruled, lined, boxed, commentsnumbered]{algorithm2e}
\usepackage{tikz}
\usetikzlibrary{shapes, arrows, positioning,fit,backgrounds,calc}
\usepackage{pgfplots}
\pgfplotsset{compat=1.16}
\usepackage{subcaption}
\usepackage{float}
\usepackage[percent]{overpic}
\usepackage[linesnumbered,ruled]{algorithm2e}
\usepackage[nodisplayskipstretch]{setspace}

\title{}
\author{}
\date{}
\allowdisplaybreaks
\title{
Real-Time Tracking in a Status Update System with an Imperfect Feedback Channel
}
\author{Saeid Sadeghi Vilni, Abolfazl Zakeri, Mohammad Moltafet, and Marian~Codreanu
\thanks{This research has been financially supported by the Infotech Oulu, the Research Council of Finland (grant 323698), and 6G Flagship program (grant 346208). M. Codreanu has also been financially supported in part by the Swedish Research Council (grant 2022-03664).\\
Saeid Sadeghi Vilni and Abolfazl Zakeri are with the CWC-RT, University of Oulu, 90014 Oulu, Finland (e-mail: Saeid.SadeghiVilni@oulu.fi; Abolfazl.Zakeri@oulu.fi). \\
Mohammad Moltafet is with the Department of Electrical and Computer Engineering, University of California at Santa Cruz, Santa Cruz, CA 95064 USA (e-mail: mmoltafe@ucsc.edu).\\
Marian Codreanu is with the Department of Science and Technology, Link\"{o}ping University, Link\"{o}ping 58183, Sweden (e-mail: marian.codreanu@liu.se).\\ \indent
The paper's preliminary results were presented at Asilomar 2024 \cite{ssvconfasl}.}}
\PassOptionsToPackage{silent}{bar}
\begin{document}

\maketitle
\begin{abstract}
    We consider a status update system consisting of a finite-state Markov source, an energy-harvesting-enabled transmitter, and a sink. The forward and feedback channels 
    are error-prone. We study the problem of minimizing the long-term time average of a (generic) distortion function subject to an energy limitation. 
    Since the feedback channel is error-prone, the transmitter has only partial knowledge about the transmission results and, consequently, about the estimate of the source state at the sink. Therefore, we model the problem as a partially observable Markov decision process (POMDP), which is then cast as a belief-MDP problem.
    The infinite belief-state space makes solving the belief-MDP difficult. Thus, by exploiting a specific property of the belief evolution, we truncate the state space and formulate a finite-state MDP problem, which is then solved using the relative value iteration algorithm (RVIA). 
    Furthermore, we propose an energy-agnostic low-complexity policy in which the belief-MDP problem is transformed into a sequence of per-slot optimization problems. Then, the energy-agnostic low-complexity policy is extended to an energy-aware low-complexity policy by adding a regularization term to the objective function of the per-slot problems.
    {Simulation results show the structure and effectiveness of the proposed policies and their superiority compared to baseline policies.}
    
   \textbf{Index Terms:} Real-time tracking, status updating, energy harvesting, partially observable Markov decision process. 
\end{abstract}
\section{Introduction}
Emerging applications in the Internet of Things (IoT) networks, such as digital health, intelligent transportation, and disaster monitoring, often rely on real-time tracking of a remotely monitored random process. 
In these systems, a sensor observes a physical process such as road traffic, patient heart rate, or temperature, and transmits real-time status updates to an intended destination, such as a remote controller or a decision-making entity \cite{uysal2022semantic}. 

A commonly used metric to quantify the information freshness in status update systems is the age of information (AoI) \cite{aoi1,ya1}. AoI is defined as the difference between the current time and the generation time of the last received status update at the destination \cite{aoi1,ya1}. Extensive research has been conducted on the AoI analysis and optimization in different areas, e.g., queuing systems \cite{ry2}, scheduling \cite{vilni2023aoi}, and sampling problems \cite{mmacm,zakeri2023minimizing}. However, AoI solely focuses on the time elapsed since the last received update's generation, neglecting the updates' content and the mismatch between the current value of the process of interest and its estimate at the destination \cite{salimnejad2024real,pp_arxiv,maatouk2022age, sun2019sampling}. 
This motivated us to address the problem of optimizing (generic) distortion measures for real-time tracking in status update systems. {The distortion function is meant to capture the discrepancy between the monitored process and its estimate at the sink relative to the application's goal. {Notably, the distortion here can be tailored to error indicator \cite{maatouk2020age}, absolute error\cite{wang2019whittle}, or squared error \cite{huang2020real}.}
}

{To make optimal decisions in a status update system, the controller must know whether the sent packet was successfully delivered.
To this end, the receiver is required to send acknowledgment feedback to the controller through a wireless channel upon a status update delivery. 
Ensuring an error-free feedback channel is challenging due to noise, interference, and fading.
To address this challenge, we study a status update system with an {imperfect} feedback channel. 
}

Furthermore, as the IoT and wireless sensor networks expand, traditional battery-powered sensors struggle to meet the growing demand for sustainability and low maintenance costs. Energy harvesting (EH) technology offers a promising alternative, allowing sensors to draw power from their environment, eliminating battery replacements, and significantly extending their lifespans \cite{abd2019role}.

In this work, we consider a status update system where an EH-enabled transmitter monitors a source and sends status updates to a sink over an unreliable wireless channel (see Fig. \ref{fsm}). The source is modeled as a finite-state Markov chain. The transmitter is equipped with a battery with limited capacity that is charged by harvesting energy from the environment. We assume a slotted communication system in which, after each successful reception, the sink sends an acknowledgment signal to the transmitter over an \textit{imperfect feedback channel}. 

{The considered system model may model status updating in an intelligent transportation system that uses EH-based roadside units \cite{atoui2018offline,yang2020analytical,gimenez2024energy}. These units harvest renewable energy sources, such as solar or wind energy, to sustain their operations. They monitor critical parameters, such as traffic flow and vehicular density, and transmit the collected data to a traffic management center via a wireless communication channel. Given the nature of such applications, the destination is interested in the real-time tracking of the monitored parameters \cite{atoui2018offline,yang2020analytical,gimenez2024energy}}.

We address a real-time tracking problem, where the goal is to minimize the expected time average of a (generic) distortion function subject to an energy limitation constraint imposed by the energy harvesting circuit. 
The solution to the problem determines the optimal transmission times. Since the feedback channel is error-prone, some of the acknowledgments may not be detected at the transmitter. Thus, the sink's estimate of the source is only partially observable at the transmitter side. Thus, we model our problem as a partially observable Markov decision process (POMDP), and subsequently, we cast it as a belief-MDP problem. Due to infinite belief space, solving the original belief-MDP problem is difficult. Thus, we truncate the state space by exploiting a specific property of the belief evolution. By truncation, we formulate a finite state belief-MDP problem, which is then solved using the relative value iteration algorithm (RVIA). Since the complexity of the POMDP-based policy increases with the number of source states and the battery's capacity, we propose two sub-optimal low-complexity (LC) transmission policies, namely: 1) energy-agnostic LC policy and 2) energy-aware LC policy. The main idea behind the energy-agnostic LC policy is that instead of solving the original belief-MDP problem, the problem is transformed into a sequence of per-slot optimization problems. When the transmitter has both options of transmitting and staying idle, the energy-agnostic LC policy makes decisions independently of the energy arrival process, while the energy arrival process affects the long-term energy availability for the transmitter's operations. Thus, by adding a regularization term to the objective function of the per-slot problems, we propose the energy-aware LC policy. The regularization term is based on the energy arrival process and promotes less transmission to save energy for future use. Specifically, the regularization term decreases the transmission rate when the energy arrival rate decreases, and its effect decreases when the energy arrival rate increases.
In the numerical results, we study the effectiveness of the proposed policies and compare the proposed policies with a baseline policy. Moreover, we study the structure of proposed policies, which reveals their switching-type structure.

The main contributions of our paper are summarized as follows:
\begin{itemize}
    \item We address a real-time tracking problem in a status update system with an \textit{imperfect feedback} channel. We investigate the problem of minimizing an average distortion function under an energy constraint imposed by the energy harvesting circuit.
    \item We model the main problem as a POMDP and propose three transmission policies: 1) a deterministic transmission policy by truncating the belief space and then applying the RVIA, 2) a sub-optimal energy-agnostic LC policy by transforming the belief-MDP problem into a sequence of per-slot problems, and 3) a sub-optimal energy-aware LC policy that takes the energy arrival rate into consideration.
    \item We numerically evaluate the proposed policies' effectiveness under different settings. We also investigate the structure of the proposed policies and show their switching-type structure. 
\end{itemize}
\subsection{Related Work}
Real-time tracking and remote state estimation have recently witnessed great interest in numerous studies, e.g., \cite{huang2020real,tsai2021unifying,zhang2023real,zhang2020error,sun2019sampling,ornee2021sampling,nadeem2022harq,wang2019whittle,wang2020real,ornee2023whittle,tang2022sampling,sun2022status,wang2020timely,yun2019information,zakeri2023optimal,nayyar2013optimal}.
The work \cite{wang2019whittle} addressed the problem of remote estimation in a multi-source status update system. Using the Whittle index theory, they proposed a heuristic channel allocation policy. This policy aims to find an optimal tradeoff between the total discounted estimation error and the energy consumption over the infinite horizon.
The authors of \cite{huang2020real} considered a hybrid automatic repeat request (HARQ) based real-time remote estimation framework for a status update system. They proposed an optimal scheduling policy to minimize the long-term mean squared error (MSE).
The work \cite{tsai2021unifying} considered a two-way delay status update system where the controller is at the sink and the source is modeled as the Wiener process. They developed a jointly optimal sensor/controller waiting policy to minimize estimation error from the sensor perspective and AoI from the controller perspective.
In \cite{zhang2023real}, the authors proposed an optimal quantization strategy for real-time monitoring of an Ornstein-Uhlenbeck (OU) process in a status update system. They minimized the time-average MSE subject to an AoI outage constraint.
The authors of \cite{zhang2020error} studied a status update system with temporally and spatially correlated source states. They used a time-varying Gauss-Markov random field to model the source state and derived the closed-form expression of the average estimation error.
The work \cite{sun2019sampling} studied a status update system where the source is modeled as a Wiener process and proposed an optimal online sampling strategy to minimize the MSE. They studied the connection between the policy that minimizes the AoI and the policy that minimizes the remote estimation error.
In \cite{ornee2021sampling}, the authors generalized the work in \cite{sun2019sampling} by considering a source modeled as OU process. They showed that the optimal sampling policy has a switching-type structure.
The authors of \cite{nadeem2022harq} investigated the real-time remote estimation in a status update system with a HARQ protocol and different retransmission policies. They proposed an optimal transmission control policy to minimize the average MSE and a variation of MSE.
In \cite{wang2020real}, the authors studied the real-time reconstruction problem of a continuous process at a remote destination. They derived the closed-form expression of the average distortion between the original and the estimated process. Moreover, they found the optimal sampling rate by minimizing the derived average distortion.
The authors of \cite{ornee2023whittle} studied the problem of minimizing the weighted sum of the expected time-average estimation errors of sources in a status update system. They proposed a sampling and transmission scheduling policy to solve the problem using the Whittle index theory.
The work \cite{tang2022sampling} studied the problem of remote estimation of a source modeled as the Wiener process over a channel with unknown delay statistics. They proposed an online sampling policy to minimize the expected time average MSE under a sampling frequency constraint.
In \cite{sun2022status}, the authors derived a closed-form expression for the average MSE under two different scheduling policies in a multi-source status update system.
The authors of \cite{wang2020timely} considered a status update system consisting of a vehicle as the source and a roadside unit as the controller. They proposed a scheduling policy to minimize the average weighted estimation error under a resource constraint.
The work \cite{yun2019information} considered a multi-source status update system and addressed the problem of joint estimation error and transmission cost minimization. They considered two threshold-based policies and evaluated them numerically.
The work \cite{zakeri2023optimal} considered an energy harvesting status update system where the source is only partially observable. They proposed several scheduling policies to minimize average distortion and age of incorrect information.


The most related works to this paper are \cite{nayyar2013optimal,rezasoltani2021real}. 
Similar to our work, the authors of \cite{nayyar2013optimal}
studied a status update system in which the source dynamic is molded via a Markov chain, and the transmitter harvests energy for operation. They considered a perfect communication channel and solved the real-time estimation problem for a finite time horizon. Different from \cite{nayyar2013optimal}, we consider imperfect forward and feedback channels and a long-term average-time problem, which makes our problem more complicated such that the solution of \cite{nayyar2013optimal} is not applicable to solve it. Notably, due to imperfect channels, status updates are not delivered for all transmission attempts, and the transmitter does not always know whether the sent status updates are successfully received.
{In \cite{rezasoltani2021real,jin2024effect}, the authors considered a status update system with an imperfect feedback channel. However, in \cite{rezasoltani2021real,jin2024effect}, the authors studied a different problem; we solve a control optimization problem to minimize the average distortion function whereas they characterized a closed-form expression of the average AoI under a fixed policy. Moreover, we consider an energy-harvesting powered transmitter, which makes the system more complicated as the transmitter does not know the availability of the energy at the subsequent slots. }



\section{System Model and Problem Formulation}\label{sysm}
\subsection{System Model}\label{sysmA}
We consider a real-time tracking system consisting of a source, an energy-harvesting powered transmitter, and a sink, as shown in Fig. \ref{fsm}. 
The transmitter monitors a stochastic process and sends status updates over an error-prone channel. 
The transmitter is powered by an energy harvesting module. The harvested energy is stored in a battery with a capacity of $B$ units of energy. 
The sink is interested in real-time tracking of the source and, upon receiving an update, sends an acknowledgment to the transmitter over an error-prone channel. The time is discrete with unit time slots ${t \in \{1, 2, . . .\}}$.
\begin{figure}
    \centering
    \includegraphics[width = 8.5cm]{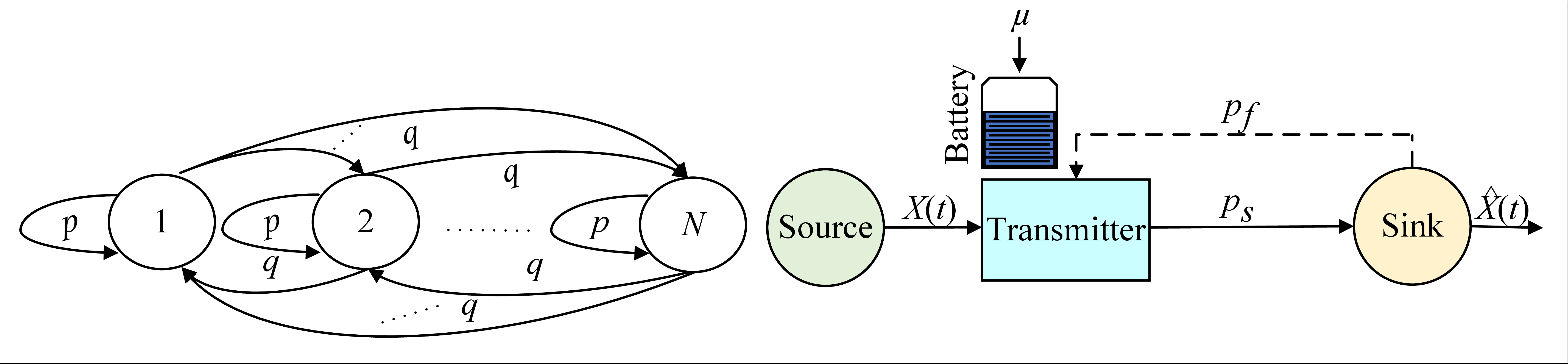}
    \caption{{The system model with the considered source model.}}
    \label{fsm}
\end{figure}

The source is modeled as a finite-state symmetric discrete-time Markov chain\footnote{This is a commonly used source model (see e.g., \cite{nayyar2013optimal,chakravorty2019remote,poojary2019real}).} and is fully observable by the transmitter. Let ${X(t)\in\{1,\ldots,N\}}$ denote the source state at slot $t$, where $N$ is the number of source states. Each slot, the source either remains in the same state with a probability $p$ or transits\footnote{We assume that the source and the system clocks are synchronized.} to any other state with a probability $ q = (1-p)/(N-1)$, as shown in Fig. \ref{fsm}. 




We denote the energy level in the battery at the beginning of slot $t$ by $b(t)\in\{0,1,\dots, B\}$.  Let  $e(t)\in\{0,1\}$ be the energy arrival process during slot $t$. We model the energy arrival as a Bernoulli process with parameter $\mu$, i.e., $\mu = \mathrm{Pr}\left(e(t)=1\right)$. Let $a(t)\in\{0,1\}$ denote the transmission decision at slot $t$, where $a(t) = 1$ indicates that the transmitter sends a status update 
and $a(t)=0$ otherwise. The transmitter spends one unit of energy for each status update transmission attempt. Therefore, the battery's energy level $b(t)$ evolves as
\begin{align}
        &b(t+1) = \min\{b(t)+e(t)-a(t),B\}.
\end{align}

The transmitter sends the status update over an error-prone wireless channel. We denote the probability of successful reception by $p_s$. After each successful reception, the sink sends an acknowledgment (ACK) to the transmitter. We assume that the feedback channel is delay-free but error-prone, i.e., ACK is either received instantaneously or lost. Let $p_f$ denote the probability of a successful ACK reception over the feedback channel. Let $f(t)\in\{0,1\}$ be the binary indicator of the reception of ACK at slot $t$; $f(t)=1$ indicates that the transmitter receives an ACK signal after a transmission attempt, and $f(t)=0$ indicates that the transmitter does not receive an ACK due to either an error over the forward channel (i.e., the status update is not received by the sink) or due to an error over the feedback channel (i.e., status update is correctly received at the sink but the ACK is lost). We also define $f(t) = 0$ for the slots where the transmitter sends no status update. 

The sink estimates the source state based on the status updates received from the transmitter. Let $\hat X(t)$ denote the estimated state of the source at slot $t$. We assume that the sink uses a maximum likelihood (ML) estimator. 
Without loss of generality, we consider the case $p> q$. In this case, Proposition~\ref{pro1} below shows that the ML source-estimator at the sink is given by the last successfully received status update.
\begin{propsn}\label{pro1}
    {Using the ML estimator, the estimate of the source state at the sink when $p > q$ is given as}
\begin{align}
 &\hat X(t+1) =
        &\begin{cases}
            X(t) & a(t)=1 ~\text{ {and the status update}}\\ &\text{{was successfully received}}\\
            \hat X(t)  & \text{{otherwise}}.
        \end{cases}
\end{align}
\end{propsn}
\begin{proof}
See appendix \ref{ap-pro1}.
\end{proof}
At each slot $t$, the transmitter decides whether to stay idle or to send a new status update based on the following available information: 1) the current source state $X(t)$, 2) the battery's energy level $b(t)$, 3) the previous source state $X(t-1)$, and 4) the ACK reception indicator at the previous slot ${f(t-1)}$.
\subsection{Problem Formulation}\label{sysmB}
{Our main goal is to minimize the average estimation error, i.e., the discrepancy between the source state $X(t)$ and its estimated value at the sink $\hat{X}(t)$, by optimizing transmission decisions subject to an energy limitation imposed by the energy harvesting circuit.} 

Let $d: \big(X,\hat X\big ) \rightarrow \mathbb{R}^+$ denote a (generic) bounded distortion function, i.e., $|d(.)|<\infty$, meant to quantify the estimation error. 
{For instance, the distortion function can be defined as the error indicator, i.e., ${d\big(X,\hat X\big )=\mathbbm{1}_{\{X\neq X\}} }$\footnote{$\mathds{1}_{\{.\}}$ is an indicator function that equals $1$ when the condition(s) in $\{.\}$ is satisfied.}, the absolute value of the difference between $X(t)$ and $\hat{X}(t)$, i.e., ${d\big(X,\hat X\big )= | X-\hat X |}$, or mean squared error, i.e., ${d\big(X,\hat X\big )= | X-\hat X |^2}$.} 
%
{Then, the considered problem can be formulated as the following stochastic control problem:}
\begin{subequations}\label{pmain}
\begin{alignat}{2}
\displaystyle\amin_{} \quad & \underset{T\rightarrow \infty}{\limisup} ~~~ \frac{1}{T} \sum_{t=1}^{T}\mathbb{E} \{ d\big(X(t),\hat X(t)\big )\}\\
\mbox{subject to} \quad & b(t)-a(t)\geq 0\label{p1:1}\\
&a(t)\in \{0,1\}\label{p1:2},
\end{alignat}
\end{subequations}
with variables $a(t)$ for all $t\in\{1,2, \dots\}$. Inequality \eqref{p1:1} represents the energy limitation constraint, where the transmitter can send a status update if the battery level is larger than zero. The expression \eqref{p1:2} indicates the binary nature of the decision variable. In problem \eqref{pmain}, $\mathbb{E}\{\cdot\}$ is the expectation with respect to the randomness of the system (i.e., randomness in the source dynamic, the energy arrival, and the forward and feedback communication channels) and the decision variable $a(t)$.

\section{Transmission Scheduling Policy}\label{policy}
In this section, we propose transmission scheduling policies to solve problem \eqref{pmain}. Note that, due to the imperfect feedback channel the source state estimate $\hat X(t)$ is only partially observable at the transmitter side. Thus, we model problem \eqref{pmain} as a partially observable Markov decision process (POMDP) which is further transformed into a belief-MDP problem. Then, we propose two scheduling policies to solve the belief-MDP problem.

\subsection{POMDP Formulation}\label{pomdpsec}
A POMDP consists of five elements: 1) states, 2) actions, 3) observations, 4) cost,  and 5) state transition probabilities, which are defined next.\\
\textbf{States}: Let ${s(t) = \big(X(t), b(t), \hat X(t),f(t-1),X(t-1)\big)\in \mathcal{S}}$ denote the state of the system at slot $t$, where $\mathcal{S}$ is the state space. Note that, due to the imperfect feedback channel, the estimated value of the source state at the sink $\hat X(t)$ is only partially observable at the transmitter side.
\\
\textbf{Actions}: Let $a(t)\in\mathcal{A}_{s(t)}$ denote the action at slot $t$, where $\mathcal{A}_s$ denotes the space of the feasible actions at state $s$, defined as $\mathcal{A}_s= \big\{a\in\{0,1\} \mid b-a\geq 0\big\}$.\\
\textbf{Observations}: Let ${o(t)= \big(X(t),b(t),f(t-1), X(t-1)\big)}$ denote the observation at slot $t$, i.e., the observable part of the state. In general, the observation function  \cite[Chapter 7]{sigaud2013markov} is defined as the probability distribution function of the current observation given the current state and previous action, i.e., ${\mathrm{Pr}(o(t)\mid s(t), a(t-1))}$. Since the observation is always part of the state, the observation function can be expressed as
\begin{align}
         &\hspace{-2.5mm}\mathrm{Pr}\big(o(t)| s(t), a(t-1)\big) = \mathbbm{1}_{\{o(t)=\left(X(t), b(t),f(t-1),X(t-1)\right)\}},
\end{align}
\textbf{Cost}: Let ${C\big(s(t)\big) = d\big(X(t),\hat X(t)\big)}$ denote the cost at slot $t$, where $d(\cdot)$ is the distortion function.\\
\textbf{State transition probabilities}: Let $\mathcal{P}(s'\mid s,a)$ denote the state transition probability defined as the probability of moving from the current state ${s}$ to the next state ${s'}$ under action ${a}$. For the sake of brevity, hereinafter, we denote $\bar \mu = 1-\mu$, $\bar p_f = 1-p_f$, $\bar p_s = 1-p_s$, and ${\tilde{b} = \min(b+1,B)}$. We further use $X_{-1}$ to denote the source state at the previous slot, $f_{-1}$ to denote the indicator of reception of ACK at the previous slot, ${s=( X, b,\hat{X},f_{-1},X_{-1})}$, and ${s'=( X', b',\hat{X}',f_{-1}',X_{-1}')}$. With these notations, the state transition probabilities can be expressed as
\begin{align}\label{prtmn}
    &\mathcal{P}(s'\mid s,a)=
    \begin{cases}
        p\bar \mu&a=0, s'=(X, b, \hat{X},0,X)\\
        p\mu&a=0, s'=(X, \tilde{b}, \hat{X},0,X)\\
        q\bar \mu&a=0, s'=(X', b, \hat{X}, 0,X)\\
        q\mu&a=0, s'=(X', \tilde{b}, \hat{X}, 0,X)\\
        p\bar \mu p_sp_f&a=1, s'=(X, b-1, X, 1,X)\\
        p \mu p_sp_f&a=1, s'=( X, b, X, 1,X)\\
        q\bar \mu p_sp_f&a=1, s'=(X', b-1, X, 1,X)\\
        q \mu p_sp_f&a=1, s'=(X', b, X, 1,X)\\
        p\bar \mu p_s\bar p_f&a=1, s'=(X, b-1, X, 0,X)\\
        p \mu p_s\bar p_f&a=1, s'=(X, b, X, 0,X)\\
        q\bar \mu p_s\bar p_f&a=1, s'=(X', b-1, X, 0,X)\\
        q \mu p_s\bar p_f&a=1, s'=(X', b, X, 0,X)\\
        p\bar \mu \bar p_s&a=1, s'=( X, b-1, \hat{X}, 0,X)\\
        p \mu \bar p_s&a=1, s'=(X, b, \hat{X}, 0,X)\\
        q\bar \mu \bar p_s&a=1, s'=(X', b-1, \hat{X}, 0,X)\\
        q \mu \bar p_s&a=1, s'=(X', b, \hat{X}, 0,X)\\
        0&\text{otherwise},
    \end{cases}
\end{align}
where $X'\in\{1,\dots,N\}$ and $ X'\neq X$.
\subsection{Belief-MDP Problem} \label{seq:Belief-MDP}


Due to the imperfect feedback channel, the estimated value of the source state at the sink $\hat X(t)$ is only partially observable at the transmitter side. Consequently, only the observable part of the state $o(t)$ is accessible for decision-making, {which is incomplete state information} and poses challenges for (optimal) transmission scheduling. To address the information insufficiency, we utilize the notion of \textit{belief} \cite[Chapter 7]{sigaud2013markov} and cast the POMDP defined in Section \ref{pomdpsec} as a belief-MDP. To this end, we define state-like quantities, referred to as belief-states, that retain the Markov property and capture all information necessary for optimal transmission scheduling. The belief-state consists of the observable component $o(t)$ of the state and a belief about the partially observable component $\hat X(t)$.

Let $I^c(t)$ denote the complete information state at slot $t$, which consists of: 1) initial probability distribution over the states, 2) the history of past and current observations, i.e., ${\{o(1),\dots,o(t)\}}$, and 3) the history of past actions, i.e., ${\{a(1),\dots,a(t-1)\}}$. Let $\rho(t)$ denote the belief about $\hat X(t)$ at slot $t$. The belief is a $N$-dimensional vector ${\rho(t)= [\rho_1(t),\dots,\rho_N(t)]}$ that {the entries are the conditional probability of the possible values of $\hat X(t)$ given $I^c(t)$.}  Formally, the $i$th entry of the belief is given as ${\rho_i(t) = \mathrm{Pr}\big(\hat X(t)=i\mid I^c(t)\big)}$.

The belief at slot $t$, $\rho(t)$, is updated based on the previous belief $\rho(t-1)$, the current observation $o(t)$, and the previous action $a(t-1)$. 
Below, we derive the belief update for three possible cases defined by the chosen action and the resulting observation.
\begin{itemize}
    \item \textbf{Case 1}: $a(t)=0$. The transmitter stays idle. 
    Since the sink does not change its estimate $\hat X(t)$ when it does not receive a status update, the belief remains unchanged, i.e.,
        \begin{align}\label{neutbe}
                &\rho(t+1) = \rho(t).
        \end{align}
        \item \textbf{Case 2}: {$a(t)=1$ and $f(t)=1$. The transmitter sends a status update containing the information $X(t)$ and receives ACK. 
        In this case, the transmitter infers the exact value of sink's estimate at slot $t+1$, i.e.,  $\hat X(t+1) = X(t)$. Thus, the entry of $\rho$ associated with $X(t)$ becomes one, and all others become zero. The belief update is given as}         
        \begin{equation}\label{case2}
        \rho(t+1)=e_{X(t)},
        \end{equation}
        where $e_n$ is the $n$th column of identity matrix $I_N$.
        {We refer to this belief as \textit{reset belief} to $X(t)$.}   
    \item \textbf{Case 3}: $a(t)=1$ and $f(t)=0$. The transmitter sends the status update $X(t)$ but does not receive ACK. This can be caused by one of the following two events: 1) event $E_1$ where the status update is received at the sink but ACK is lost over the feedback channel, and 2) event $E_2$ where the sink does not receive the status update due to an error over the forward channel. Under $E_1$, sink's estimate at slot $t+1$ become $\hat X(t+1) = X(t)$; under $E_2$, sink's estimate remains unchanged, i.e., $\hat X(t+1) = \hat X(t)$. Thus, the belief update can be written as
   \begin{equation}\label{case3apx}
  \rho(t+1) = \mathrm{Pr}(E_1\mid \overline{\text{ACK}}) e_{X(t)} +  \mathrm{Pr}(E_2\mid \overline{\text{ACK}}) \rho(t) 
  \end{equation}
  where we used $\overline{\text{ACK}}$ to denote the lack of ACK event and the fact that given $\overline{\text{ACK}}$, events $E_1$ and $E_2$ are mutually exclusive and collectively exhaustive. 

We show in Appendix~\ref{case3p1} that the conditional probabilities in~(\ref{case3apx}) are given by
\begin{align}\label{case3e1}
\mathrm{Pr}(E_1\mid\overline{\text{ACK}}) & = \frac{p_s(1-p_f)}{1-p_sp_f}, \\
\label{case3e2} 
\mathrm{Pr}(E_2\mid\overline{\text{ACK}}) & = \frac{1-p_s}{1-p_sp_f}.
\end{align}

By substituting \eqref{case3e1} and \eqref{case3e2} into \eqref{case3apx}, the belief update for Case 3 can be expressed as
\begin{equation}\label{case3apxfinal}
  \rho(t+1) = \frac{p_s(1-p_f)}{1-p_sp_f} e_{X(t)} +  \frac{1-p_s}{1-p_sp_f} \rho(t). 
\end{equation}
\end{itemize}
The belief-MDP is defined by the following elements:\\
\textbf{Belief-states}: Let $l(t) \!=\! \big(X(t),b(t),\rho(t),f(t\!-\!1),X(t\!-\!1)\big)\!\in \! \mathcal{L}$ denote the belief-state at slot $t$, where $\mathcal{L}$ is the belief-state space.\\
\textbf{Actions}: Let $a(t)\in\mathcal{A}_l$ denote the action at slot $t$, where $\mathcal{A}_l$ is the space of the feasible actions at belief-state $l$.
Note that since the action space depends on the observable part of the belief-state $o(t)$, the action space in the POMDP and the belief-MDP are the same.\\
\textbf{Cost}: Let $C\big(l(t)\big)$ denote the cost at slot $t$ defined as the expected value of the distortion function with respect to the belief, given as
\begin{align}\label{bmcost}
         &  C\big(l(t)\big) = \sum_{i=1}^{N}\rho_i(t)d\big(X(t),i\big).
\end{align}
\textbf{Belief-state transition probabilities}: Let ${\mathcal{P}(l'\mid l,a)}$ denote the state transition probability defined as the probability of moving from the current belief-sate ${l = \big(X, b, \rho,f_{-1},X_{-1}\big)}$ to the next belief-state ${l' = \big(X', b', \rho',f_{-1}',X_{-1}'\big)}$ under action $a$. For the sake of brevity, we use $\tilde\rho$ to denote the value of the next belief for the case $a(t)=1$ and $f(t)=0$, obtained from \eqref{case3apxfinal}. 
The belief-state transition probabilities can be expressed as
\begin{align}\label{prtm}
    &\mathcal{P}(l'\mid l,a)=\notag\\&\hspace{-1mm}
    \begin{cases}
        p\bar \mu&a=0, l'=(X, b, \rho, 0,X)\\ 
        p \mu&a=0, l'=(X, \tilde{b}, \rho, 0,X)\\
        q\bar \mu&a=0, l'=(X', b, \rho, 0,X)\\  
        q \mu&a=0, l'=(X', \tilde{b}, \rho, 0,X)\\
        p\bar \mu p_sp_f&a=1, l'=(X, b-1, e_X, 1,X)\\ 
        p\mu p_sp_f&a=1, l'=(X, b, e_X, 1,X)\\
        q\bar \mu p_sp_f&a=1, l'=(X', b-1, e_X, 1,X)\\ 
        q\mu p_sp_f&a=1, l'=(X', b, e_X, 1,X)\\
        p\bar \mu (1-p_sp_f)&a=1, l'=(X, b-1, \tilde\rho, 0,X)\\ 
        p\mu (1-p_sp_f)&a=1, l'=(X, b, \tilde\rho, 0,X)\\
        q\bar \mu (1-p_sp_f)&a=1, l'=(X', b-1, \tilde\rho, 0,X)\\ 
        q\mu (1\!-\!p_sp_f)&a=1, l'=(X', b, \tilde\rho, 0,X)\\
        0&\text{otherwise},
    \end{cases}
\end{align}
where $X'\in\{1,\dots,N\}$ and $ X'\neq X$. 
Let $\pi$ represent a policy that determines the action to be taken at each belief-state. {A randomized policy is determined by a distribution ${\pi(a\mid l): {\mathcal{L}\times \mathcal{A}_l, } \rightarrow [0,1]}$}. A special case is a deterministic (stationary) policy that associates probability one to a single action at any given belief-state. To simplify notation, we use $\pi(l)$ to denote the action taken in belief-state $l$ under a deterministic policy $\pi$.

Given an initial belief-state $l(0)$, the belief-MDP problem is given as
\begin{align}\label{pomdpp}
        \amin_{\pi\in \Pi} \quad  \underset{T\rightarrow \infty}{\limisup} \frac{1}{T} \sum_{t=1}^{T}\mathbb{E} \{ C\big(l(t)\big)\mid l(0)\},
\end{align}
where $\Pi$ is the set of all feasible policies. We denote by $C^*\big(l(0)\big)$ the optimal value of problem \eqref{pomdpp} and 
by $\pi^*\big(l(0)\big)$ an optimal policy.

Note that, for a general MDP, the optimal value $C^*$ and the optimal policies $\pi^*$ may depend on the initial state $l(0)$. 
However, Theorem~\ref{theo1} below shows that the belief-MDP~\eqref{pomdpp} is communicating \cite[Chapter 8.3.1]{puterman1994} and, therefore,  its optimal value $C^*$ is independent of the initial state. 

\begin{theorem}\label{theo1}
    The belief-MDP 
    defined above is communicating, i.e., for any two states $l$ and $l'$ in $\mathcal L$, $l'$ is accessible from $l$.
\end{theorem}
\begin{proof}
We show that, starting from any initial state $l \in \mathcal{L}$, a simple randomized policy which sends updates with probability $0.5$ whenever there is energy in the battery, reaches any other state $l'\in \mathcal{L}$ with a positive probability. The details are presented in Appendix \ref{ap_p_th}.
\end{proof}

\begin{corollary}
The optimal value of problem~\eqref{pomdpp}  is the same for all initial states, i.e., $C^*(l) = C^*$ for all $l \in \mathcal{L}$ .
\end{corollary}
\begin{proof}
Follows directly from \cite[Proposition 4.2.3]{bertks}.
\end{proof}

\begin{corollary}[Bellman's Equation]
If there is a scalar $\bar h$ and a set of values $\{h(l)\}_{l \in \mathcal{L}}$ that satisfy 
 \begin{align}\label{eq:Bellman}
 &\hspace{-2.1mm}\bar h \!+\! h(l) \!= \!\min_{a\in \mathcal{A}_l}\left[C(l) + \sum_{l'\in \mathcal{L}}\mathcal{P}(l' \mid l,a)h(l')  \right]\!,~ \forall l\in \mathcal{L},
 \end{align}
then $\bar h$ is the optimal value, i.e., $\bar h = C^*$; furthermore,  the actions $a_l^*$ that attain the minimum in \eqref{eq:Bellman} for each $l$, provide an optimal deterministic policy $\pi^*$. 
\end{corollary}
\begin{proof}
Follows directly from \cite[Proposition 4.2.3]{bertks}.
\end{proof}

Note that Bellman's equation~\eqref{eq:Bellman} represents a system of $|\mathcal{L}|$ (nonlinear) equations with $|\mathcal{L}|+1$ unknowns. Thus, its solution is not unique, i.e., if  $\{h(l)\}_{l \in \mathcal{L}}$ satisfy \eqref{eq:Bellman} then so does $\{h(l) + k \}_{l \in \mathcal{L}}$ for any $k$. This suggests that an ordinary value iteration (commonly used for solving discounted MDPs), which updates \emph{all} values $\{h(l)\}_{l \in \mathcal{L}}$ at every step, will likely fail to converge here. Thus, in the next section, we propose a relative value iteration method which selects an arbitrary reference state $l_{\text{ref}} \in \mathcal{L}$, sets its corresponding entry in $h(l)$ to zero, i.e.,   $h(l_{\text{ref}}) = 0$, and each step updates all others entries  $\{h(l)\}_{l \in \mathcal{L}, l \neq l_{\text{ref}}}$. 

\subsection{POMDP-based Solution to Belief-MDP Problem}\label{trpomdp}
The relative value iteration algorithm (RVIA)~\cite[Section~4.3]{bertks} turns  Bellman's equation~\eqref{eq:Bellman} into an iterative procedure which selects an arbitrary reference state  $l_{\text{ref}}\in\mathcal{L}$, kept fixed throughout the iterations, and each step $i$ updates the estimates $\{h^i(l)\}_{l \in \mathcal{L}, l \neq l_{\text{ref}}}$ of the relative value function, until convergence. Specifically, we initialize $h^0(l) = 0$ for all $l \in \mathcal{L}$ and for $i=1,2,\ldots$, we update
 \begin{align} \label{RVIA_V_update}
V^{i}(l) & = \min_{a\in \mathcal{A}_l}\left[C(l) +\sum_{l'\in \mathcal{L}}\mathcal{P}(l' \mid l,a)h^{i-1}(l')\right], 
\\ \label{RVIA_h_update}
h^{i}(l) & = V^{i}(l) - V^{i}(l_{\text{ref}}),
\end{align}
until $\displaystyle \max_{l\in\mathcal{L}}|h^{i}(l)-h^{i-1}(l)|\leq\epsilon$, where $\epsilon$ is small constant used for the stopping criterion.  
Once the iterative process of RVIA converges, i.e., $\displaystyle h(l) = \lim_{i\rightarrow\infty} h^i(l)$, the algorithm provides an optimal deterministic policy $\pi^* (l) = a^*_l$, where $\{a^*_l\}_{l\in \cal{L}}$ is the set of actions that attain the minimum in~(\ref{RVIA_V_update}).  Furthermore, the optimal value of problem \eqref{pomdpp} is given by $\displaystyle C^{*} = \lim_{i\rightarrow\infty} V^i(l_{\mathrm{ref}})$.

To employ RVIA, the belief-state space (and the action space) must be finite. However, since the believe entry $\rho(t)$ is continuous, i.e., $\rho(t) \in \mathbb{R}$, the belief-state space $\cal{L}$ become infinite. Fortunately, the belief evolution follows a specific pattern, which we exploit to truncate the belief-state space $\cal{L}$ and subsequently to develop a practical approximate solution using RVIA. 

{To express a belief space compactly let $u_1$ and $u_2$ denote the probabilities in~\eqref{case3apx} as ${u_1=\mathrm{Pr}(E_1\mid \overline{\text{ACK}})}$ and ${u_2=\mathrm{Pr}(E_2\mid \overline{\text{ACK}})}$.} 
Note that, whenever the transmitter sends a status update and receives ACK, the belief is set to a specific value given by~(\ref{case2}) (i.e., $e_{X(t)}$).
Furthermore, the belief remains unchanged whenever the transmitter stays idle (see (\ref{neutbe})). Thus, from~\eqref{case3apx}, the space of all possible values that $\rho(t)$ can reach during $m$ consecutive transmission attempts without receiving ACK can be expressed as $\Omega_m = \cup_{i=1}^{m}\Phi(i)$,
where
\begin{align}\label{phimn}
   \Phi(i)=&\{u_2^i x_i+ u_2^{i-1}u_1 x_{i-1}+ \cdots + u_1 x_{0}\mid \notag\\ &x_0,\cdots,x_i \in \{e_1,\ldots,e_N \}\},
\end{align}
i.e., linear combinations between all possible sequences of ${m+1}$ standard basis vectors with the coefficients specified above.





To obtain a finite belief-state space $\cal L$, we restrict all the possible values of the belief $\rho(t)$ to those contained in ${\Omega}_m$.  This is reasonable because the probability of not receiving feedback during $m$ consecutive transmission attempts is $(1-p_sp_f)^m$, i.e., it decreases exponentially in $m$. Thus,  for a large enough $m$, the space of all possible values of $\rho(t)$ is well approximated by ${\Omega}_m$. In the unlikely (but still possible) event when the number of consecutive unsuccessful transmission attempts exceeds $m$, we replace the value of the belief $\rho(t)$ obtained from~\eqref{case3apx} by the closest value contained in set ${\Omega}_m$.
Closeness is measured using the Kullback–Leibler (KL) divergence \cite[Chapter 2]{cover1999elements}. We define the updated belief $\rho^*$ as
\begin{equation}
\rho^{*} \in \argmin_{\rho \in {\Omega}_m} ~~~\operatorname{KL}\left(\rho(t)|| \rho\right)    
\end{equation}
where the KL divergence between the beliefs is given by $\operatorname{KL}\left(\rho(t)|| \rho\right)  = \sum_{i=1}^N \rho_i(t)\log \frac{\rho_i(t)}{\rho_i}$.
This operation can be implemented efficiently via an offline generated look-up-table which provides the closest belief inside ${\Omega}_m$ for each element in $\Phi(m+1)$. 


{Replacing the beliefs of the belief-states associated with $\Phi(m+1)$ with the closest beliefs inside $\Omega_m$  
adds new (possible) transitions between belief-states in the truncated belief-state space. Thus, we shall update the belief-state transition probabilities, i.e., $\mathcal{P}(l'\mid l,a)$ given in \eqref{prtm}, with a new state transition probabilities denoted as  $\widetilde {\mathcal{P}}(l'\mid l,a)$ for all $l'$ and $l$ in the truncated belief-state space.
A transition from $l\in\mathcal{L}$ to $l'\in\mathcal{L}$ could possibly be directly, or
indirectly by the projection, i.e., 
$l$ transits to
a belief-state outside of the truncated belief-state space whose projection is $l'$. 
Thus,
}
\begin{align}\label{eq:uppr}
\displaystyle
    \widetilde {\mathcal{P}}(l'\mid l,a) = \mathcal{P}(l'\mid l,a) + \sum_{\underline l\in\underline{\mathcal{L}}}{\Pr}\{\underline{l}\,|\,l,a\},
\end{align}
{where $\underline{\mathcal{L}}$ is a set of belief-states that are outside of the truncated belief-state space, and their projection is $l'$.}

{Using truncated belief space and the transition probabilities in \eqref{eq:uppr}, we implement RVIA to solve the belief-MDP problem. The details of the proposed POMDP-based transmission policy are presented in Algorithm~\ref{pomdpa}.
}
\begin{algorithm}
    \SetKwInOut{Inputi}{Initialize}
    \SetKwInOut{run}{RUN}
     \SetKwInOut{output}{Output}
     \SetKwInOut{Output}{Output}
     \SetKwComment{Comment}{/*}{ }
     \SetKwRepeat{Do}{do}{while}
    \Inputi{ $l_{\mathrm{ref}}$, $\epsilon$, $i=0$, set $h^0(l) =0$ for all $l\in\mathcal{L}$}
    \Do{$\displaystyle\max_{l\in\mathcal{L}}|h^{i}(l)-h^{i-1}(l)|\geq\epsilon$}{
     $i = i+1$\\
    \For{$
    \displaystyle
    l\in\mathcal{L}$}{
    $\displaystyle V^{i}(l) = \min_{a\in \mathcal{A}_l}\left[C(l) +\sum_{l'\in \mathcal{L}}\widetilde {\mathcal{P}}(l'\mid l,a)h^{i-1}(l')\right]$\\
    $h^{i}(l) = V^{i}(l)-V^{i}(l_{\mathrm{ref}})$\\
    } }
    \Comment{Generate a deterministic transmission policy}
    
    $\displaystyle \pi(l)=\argmin_{a\in \mathcal{A}_l}\left[C(l) +\sum_{l'\in \mathcal{L}}\widetilde {\mathcal{P}}(l'\mid l,a)h(l')\right],~{\forall l\in \mathcal{L}}$
 
    \KwOut{POMDP-based transmission policy $ \pi^*=\pi$, the optimal value $C^* = V(l_{\mathrm{ref}})$} 
    \caption{The POMDP-based transmission policy}
    \label{pomdpa}
\end{algorithm}

\subsection{Low-complexity Solution to the Belief-MDP Problem}\label{lxp}
The POMDP-based policy, presented in Algorithm \ref{pomdpa}, needs to explore all states and actions, and the size of the belief-state space grows exponentially with the number of source states and the battery's capacity. Therefore, in the following, we introduce two low-complexity (LC) policies.
\subsubsection{Energy-Agnostic LC Policy}
Inspired by the  Lyapunov drift-plus penalty method \cite{lyp}, we develop an LC dynamic transmission policy in which the belief-MDP problem is transformed into a sequence of per-slot optimization problems. In particular, at each slot $t$, the aim is to solve the following problem
\begin{equation}\label{Problem: low-complex_main}
        \amin_{a(t)\in\mathcal{A}_l} \quad  \mathbb{E} \{ C(t+1)\mid l(t)\},
\end{equation}
where ${\mathbb{E} \{ C(t+1)\mid l(t)\}}$ is the (conditional) expected cost at slot $t+1$ given by \eqref{Eq:expex}, whose derivation is presented in Appendix \ref{A:expex1}.
\begin{align}\label{Eq:expex}
    &\mathbb{E}\{C(t+1)\mid l(t)\} = d\big(X(t),X(t)\big)pp_s\mathbb{E}\{a(t)\mid l(t)\}\notag\\&+\sum_{i=1}^{N}d\big(X(t),i\big)\rho_i(t)p(1-p_s)\mathbb{E}\{a(t)\mid l(t)\}\notag\\&+ \sum_{j=1,j\neq X(t)}^{N}d\big(j,X(t)\big)qp_s\mathbb{E}\{a(t)\mid l(t)\} \notag\\&+ \sum_{i=1}^{N}\sum_{j=1,j\neq X(t)}^{N}\hspace{-0mm}d(j,i)\rho_i(t)q(1\!-\!p_s)\mathbb{E}\{a(t)\mid l(t)\} \notag\\&+ \sum_{i=1}^{N}d\big(X(t),i\big)\rho_i(t)p(1-\mathbb{E}\{a(t)\mid l(t)\})\notag\\&+  \sum_{i=1}^{N}\sum_{j=1,j\neq X(t)}^{N}d(j,i)\rho_i(t)q(1-\mathbb{E}\{a(t)\mid l(t)\}).  
\end{align}
We minimize \eqref{Eq:expex} following the approach of opportunistically minimizing a (conditional) expectation \cite{lyp}, i.e., \eqref{Eq:expex} is
minimized by dropping the expectations in each slot.
Since there are only two possible actions at each slot, the exhaustive search method is an appropriate method to solve the per-slot optimization problem. {If the resulting expected cost for both actions are the same, the action of staying idle, i.e., $a=0$, is selected.} Note that this policy does not require the truncation of the belief-state space.

The above proposed policy is agnostic to the energy arrival process, while the energy arrival process influences the long-term availability of energy for the transmitter's operation. Thus, ignoring the energy arrival rate, especially when the energy arrival rate is low, would lead to inefficient decision-making. To address this shortcoming, we propose an energy-aware LC policy, as detailed below.
\subsubsection{{Energy-Aware LC Policy}}\label{Ilxp}
To take the energy arrival process into account, we introduce a regularization term to add to the objective function of Problem \eqref{Problem: low-complex_main}. The regularization term is a function of the energy arrival process and is meant to promote less transmissions to save energy for future use. More precisely, the regularization term decreases the transmission rate when the energy arrival rate decreases, while its effect diminishes as the energy arrival rate increases. 
This term is defined as $-\gamma a(t)\big(e(t)-1\big)$, where $\gamma$ is a positive tuning parameter. Thus, the new per-slot optimization problem is formulated as
\begin{equation}\label{Problem: low-complex_main2}
        \amin_{a(t)\in\mathcal{A}_l} \quad  \mathbb{E} \{ C(t+1)-\gamma a(t)\big( e(t)-1 \big)\mid l(t)\},
\end{equation}
The conditional expectation ${\mathbb{E} \{-\gamma a(t)\big( e(t)-1 \big)\mid l(t)\}}$ is given as
\begin{align}\label{Eq:expex22}
    &\hspace{-1.5mm}\mathbb{E}\{-\gamma a(t)\big( e(t)-1 \big)\mid l(t)\}=- \gamma\mathbb{E}\{a(t)\mid l(t)\}(\mu-1),  
\end{align}
{where the equality followed from the fact that ${\mathbb{E}\{ e(t)\mid l(t)\}=\mu }$, and the action decision at slot $t$ is independent of the energy arrival during the slot $t$.
Thus, substituting \eqref{Eq:expex} and \eqref{Eq:expex22} into \eqref{Problem: low-complex_main2}, the expected cost function is given as}
\begin{align}\label{Eq:expex23}
    &\mathbb{E} \{ C(t+1)-\gamma a(t)\big( e(t)-1 \big)\mid l(t)\} \notag\\&= d\big(X(t),X(t)\big)pp_s\mathbb{E}\{a(t)\mid l(t)\}\notag\\&+\sum_{i=1}^{N}d\big(X(t),i\big)\rho_i(t)p(1-p_s)\mathbb{E}\{a(t)\mid l(t)\}\notag\\&+ \sum_{j=1,j\neq X(t)}^{N}d\big(j,X(t)\big)qp_s\mathbb{E}\{a(t)\mid l(t)\} \notag\\&+ \sum_{i=1}^{N}\sum_{j=1,j\neq X(t)}^{N}\hspace{-0mm}d(j,i)\rho_i(t)q(1\!-\!p_s)\mathbb{E}\{a(t)\mid l(t)\} \notag\\&+ \sum_{i=1}^{N}d\big(X(t),i\big)\rho_i(t)p(1-\mathbb{E}\{a(t)\mid l(t)\})\notag\\&+  \sum_{i=1}^{N}\sum_{j=1,j\neq X(t)}^{N}d(j,i)\rho_i(t)q(1-\mathbb{E}\{a(t)\mid l(t)\})\notag\\&
    - \gamma\mathbb{E}\{a(t)\mid l(t)\}(\mu-1).  
\end{align}

We follow the same approach to minimize \eqref{Eq:expex22} as we employed to minimize \eqref{Eq:expex}. In particular, the energy-agnostic LC policy is a special case of the energy-aware LC policy when the tuning parameter is set to zero, i.e., $\gamma=0$.

\section{Numerical Results}\label{snrs}
In this section, we evaluate the performance of the proposed policies: 1) {the} POMDP-based policy {presented in Algorithm~\ref{pomdpa}} and 2) {the} LC policies presented in Section \ref{lxp}. In the simulation results, we use the following distortion function: ${ d\big(X(t),\hat{X}(t)\big) = |X(t)-\hat{X}(t)| }$.

In the following, we first study the effect of truncation on the performance of the POMDP-based policy, then we study the structure of the proposed policies, and finally, we analyze the performance of the proposed policies with respect to different system parameters.

\subsection{POMDP-based Policy}
In Fig. \ref{fig:statespace}, we evaluate the effect of belief-state truncation on the performance of the POMDP-based policy (presented in Algorithm \ref{pomdpa}) with $\epsilon=1e-4$. Fig.~\ref{fig:statespace} depicts the average cost versus the number of consecutive transmission attempts without receiving ACK, i.e., $m$, for different values of probabilities of successful reception in the forward and backward channels. As can be seen, when the channels' condition is bad, e.g., $p_s=p_f=0.2$, the performance of the derived policy converges slowly with $m$. This is because, with a small value of $p_s$ and $p_f$, the number of consecutive transmission attempts without receiving ACK increases, and thus, the truncated belief space needs to include states associated with large values of $m$ to be a good approximation for the actual belief space.
 However, when the probability of receiving ACK increases, the 
 average cost $\bar{C}$
 converges quickly with $m$.
Fig.~\ref{fig:statespace} further shows that increasing $m$ beyond the value $m=6$ leads to negligible improvement in the objective function. Hence, a value of $m=6$ is a reasonable choice to truncate the belief-state space for the considered setting.
 
 
\begin{figure}[!ht]
\centering
\begin{tikzpicture}[scale=1]
\begin{axis}[
    title={},
    xlabel={Number of slots without receiving ACK, $m$},
    ylabel={Average cost, $\bar C$},
    xmin=1, xmax=7,
    ymin=0.55, ymax=0.76,
    xtick={1,2,3,4,5,6,7},
    ytick={0.60,0.7,0.80},
    legend pos=north west,
    ymajorgrids=true,
    grid style=dashed,
    legend style={at={(0.35,0.85)},anchor=west},
    legend cell align=left,
]

\addplot[
    color=blue,
    mark=square,
    line width=1.3pt, 
    mark size=3pt,
    ]
    coordinates {
    (1,0.7565)(2,0.7433)(3,0.6898)(4,0.6833)(5,0.6806)(6,0.6783)(7,0.6768)
    };

    \addplot[
    color=red,
    mark=diamond,
    line width=1.3pt, 
    mark size=3pt,
    ]
    coordinates {
    (1,0.6667)(2,0.6654)(3,0.6648)(4,0.6579)(5,0.6553)(6,0.6541)(7,0.6538)
    };

        \addplot[
    color=black,
    mark=o,
    line width=1.3pt, 
    mark size=3pt,
    ]
    coordinates {
    (1,0.5714)(2,0.5664 )(3,0.5660 )(4,0.5657)(5,0.5656)(6,0.5655)(7,0.5655)
    };
    
    \legend{$p_s=0.2$ $p_f=0.2$, $p_s=0.4$ $p_f=0.4$, $p_s=0.6$ $p_f=0.6$}
\end{axis}
\end{tikzpicture}
  \caption{The average cost $\bar C$ versus $m$ for different reliabilities of the forward and feedback channels where the source is characterized as $p = 0.7$, $N = 3$ and the EH module is characterized as $\mu = 0.5$, $B = 3$.}\label{fig:statespace}
  \end{figure}

Fig. \ref{fig:pomdpiteration} illustrates the evolution of the average cost $\bar{C}$ as a function of time for different probabilities of successful reception $p_s$ under POMDP-based policy. From Fig. \ref{fig:pomdpiteration}, we see that the average cost decreases when $p_s$ increases. This is due to the fact that by improving the channel condition the number of successful receptions increases, and consequently, the number of slots at which the system is in an erroneous state decreases.
\begin{figure}
    \centering
    \includegraphics[width =9cm]{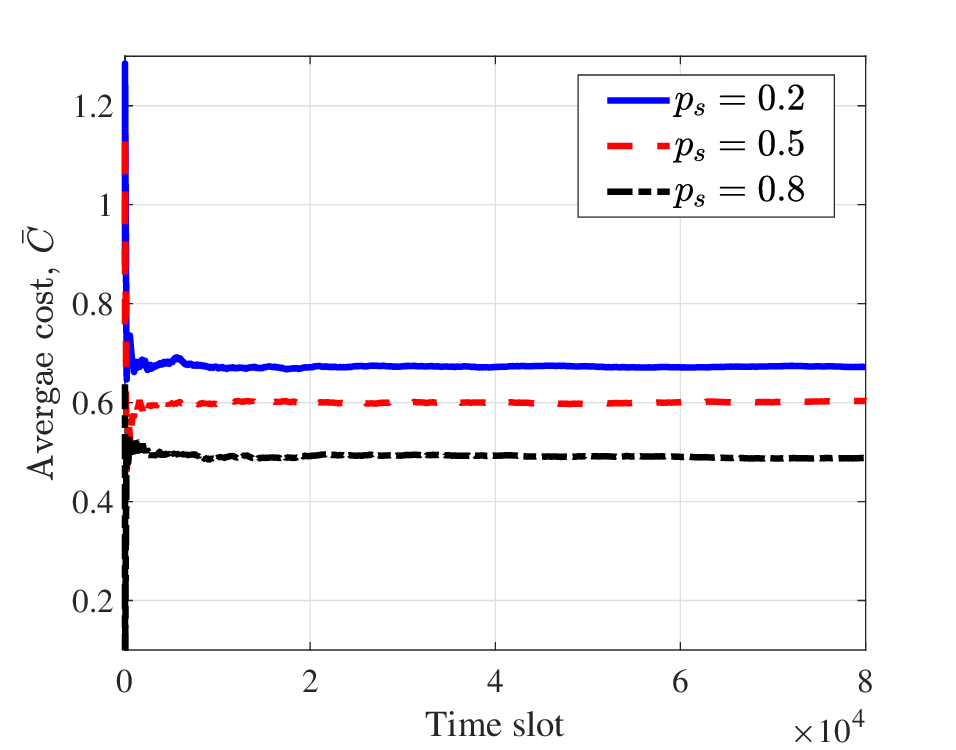}
    \caption{The (instantaneous) average cost $\bar C$ for different values of successful reception in the forward channel for the POMDP-based policy with respect to time slots where $p = 0.7$, $N = 3$, $\mu = 0.5$, $B = 3$, and $p_f = 0.7$.}
    \label{fig:pomdpiteration}
\end{figure}
\subsection{The Structure of Proposed Policies}
Fig. \ref{fig:streh} visualizes the structure of the POMDP-based policy for different values of the energy arrival rates $\mu$ where the source is binary, and the source state is $X=1$.
The plot is drawn with respect to the battery's energy level $b$, and the belief associated with $\hat X = 1$, i.e., $\rho_1$. As can be seen, the policy has a threshold structure\footnote{{A threshold structure in a policy refers to a decision rule where different actions are taken before and after a certain threshold.}} with respect to
$b$ and $\rho_1$. For example, when $\mu = 0.2$ and $\rho_1 = 0.5598$, the action decision for $b\geq3$ is transmission $a=1$ and for $b<3$ is staying idle $a=0$. 
In addition, Fig. \ref{fig:streh} shows that by increasing $\mu$, the number of belief states that the policy decides to transmit increases. 
This is because increasing $\mu$ results in having more energy in the system. 

\begin{figure}
    \centering
    \includegraphics[width =9cm]{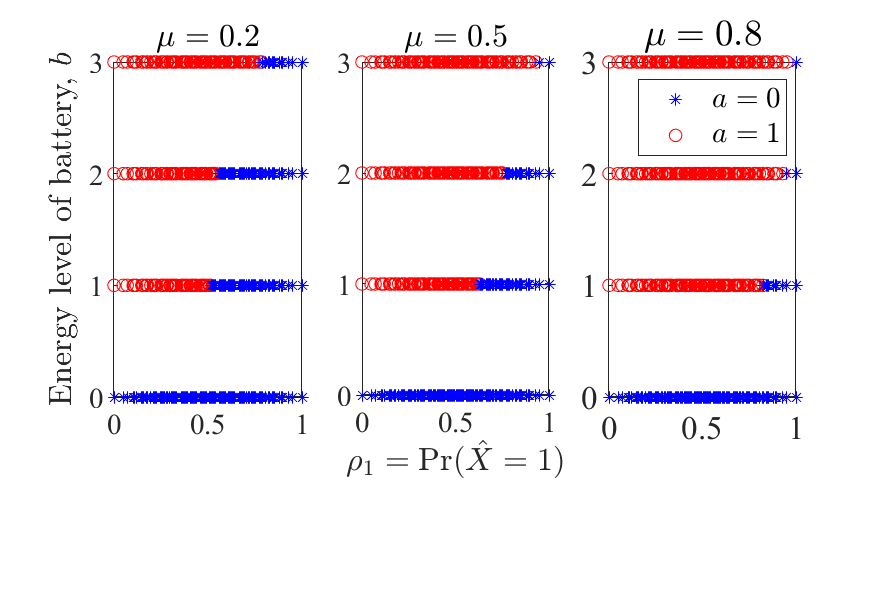}
    \vspace{-1.2cm}
    \caption{Structure of the POMDP-based policy for a binary source with respect to the belief and energy level of the battery $b$ for three different energy arrival rate $\mu$, where $X=1$, $p = 0.7$, $N = 2$,  channel reliabilities are defined as $p_s=0.6$, $p_f = 0.7$, and  $B = 3$.}
    \label{fig:streh}
\end{figure}

Fig. \ref{fig:strsd} visualizes the structure of the POMDP-based policy for different source self-transition probabilities $p$, where the source is binary, and the source state is $X=1$. The plot is drawn with respect to $b$ and $\rho_1$. The figure shows that when $p$ is small, e.g., $p= 0.5$, the policy tends not to transmit. This is because there is a high chance that the source would change its state at the next slot, and the currently transmitted status update cannot decrease the average distortion at the next slot. However, when $p$ increases, the policy tends to have transition attempts more frequently. It can also be seen that the policy has a threshold structure with respect to $b$ and $\rho_1$. For instance, when $p = 0.9$ and $\rho_1 = 0.8594$, the threshold value for the battery's energy level is $b = 2$. Moreover, the results show that when $p=q$ (e.g., $p=0.5$ and $q=0.5$), in line with \cite{maatouk2020age}, the optimal decision is to stay idle.
\begin{figure}
    \centering
    \includegraphics[width =9cm]{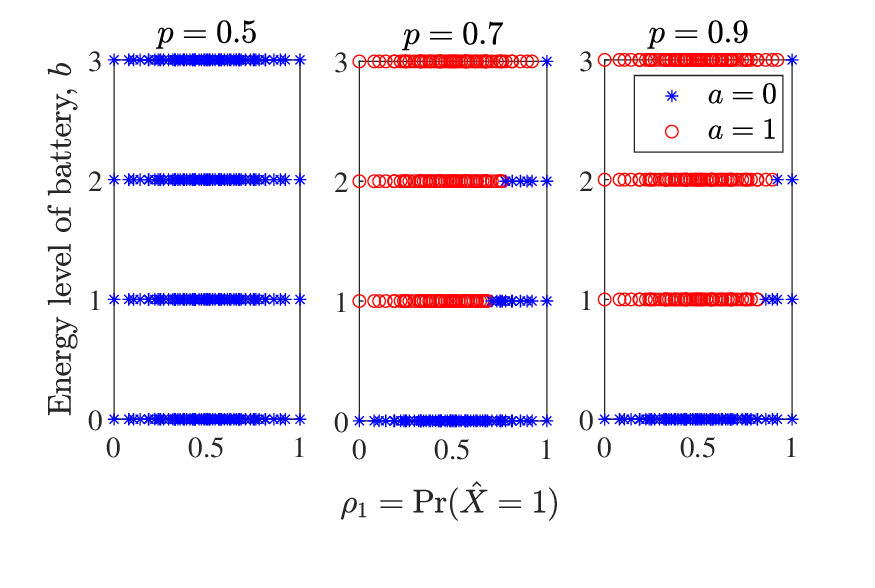}
    \vspace{-0.8cm}
    \caption{Structure of the POMDP-based policy for a binary source with respect to the belief and energy level of the battery $b$ for three different source's self-transition probability $p$, where $X=1$, $N = 2$, $p_s=0.4$, $p_f = 0.5$, $\mu = 0.7$, and $B = 3$.}
    \label{fig:strsd}
\end{figure}

Fig. \ref{fig:comstrlcp} demonstrates the structure of the LC policies for different source self-transition probabilities $p$, where the source is binary, and the source state is $X=1$. The plot is drawn with respect to $b$ and $\rho_1$. { Fig. \ref{fig:comstrlcp}(a) shows the structure of the energy-agnostic LC policy, and Fig. \ref{fig:comstrlcp}(b) shows the structure of the energy-aware LC policy.
As can be seen, the decision is independent of $b$ and changes with $p$ for both policies. For small values of $p$, the decision determined by the policies is to stay idle, while for a larger value of $p$, the decision is to send updates, which is the same as the decision determined by the POMDP-based policy (see Fig.~\ref{fig:strsd}). For $p>0.5$, the energy-agnostic LC policy sends updates when there is energy in the battery unless the expected distortion cost is zero, i.e., $\rho_1=1$. However, the energy-aware LC policy stays idle when the expected distortion cost is small (e.g., $\rho_1\geq0.7$). This is because the distortion cost is small, sending new updates does not reduce the cost significantly.} 



\begin{figure}[h]
    \centering
    \begin{subfigure}[t]{0.45\textwidth}
        \centering
        \includegraphics[width=\textwidth]{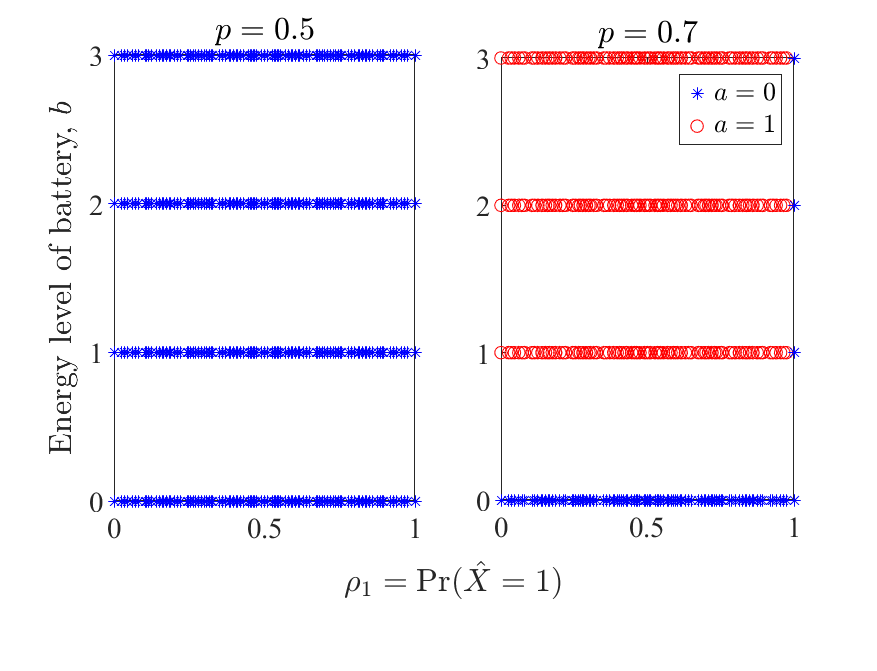}
        \vspace{-1.1cm}
        \caption{Energy-agnostic LC policy}
        \label{fig:subfig1}
    \end{subfigure}
    \hfill
    \begin{subfigure}[t]{0.45\textwidth}
        \centering
        \includegraphics[width=\textwidth]{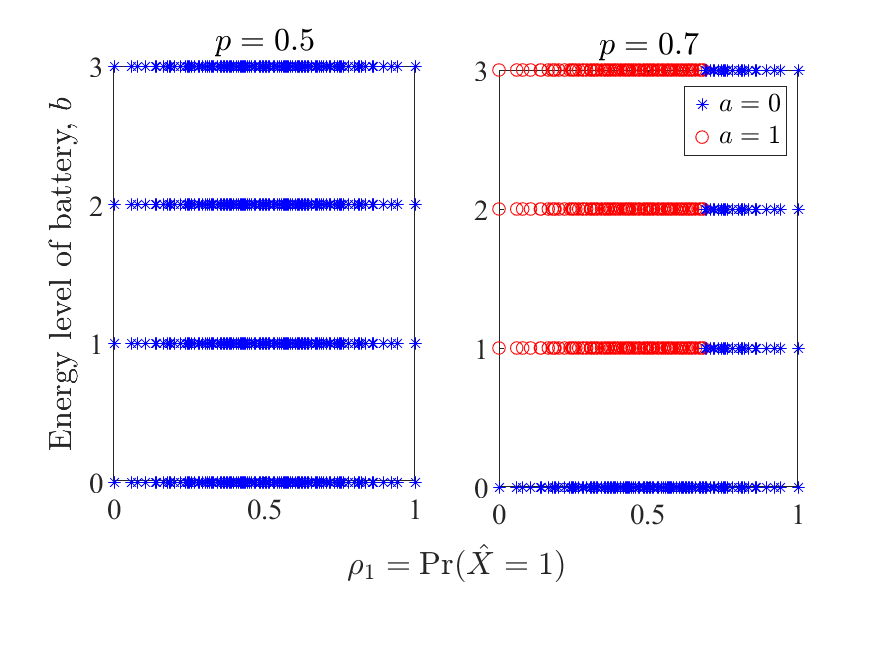}
        \vspace{-1.1cm}
        \caption{Energy-aware LC policy}
        \label{fig:subfig2}
    \end{subfigure}
    \caption{Structure of the LC policies for a binary source with respect to the belief and energy level of the battery $b$ for three different source's self-transition probability $p$, where $X=1$, $N = 2$, $p_s=0.6$, $p_f = 0.5$, $\mu = 0.7$, and $B = 3$.}
    \label{fig:comstrlcp}
\end{figure}

\subsection{Performance Comparison of the Proposed Policies}
In this subsection, we compare the proposed policies, i.e., the POMDP-based policy presented in Section \ref{trpomdp} and LC policies presented in Section \ref{lxp}, against each other and two baseline policies. According to the first baseline, the transmitter sends an update whenever the battery's energy level is not zero, referred to as battery-only (BO) policy.
In the second baseline, the transmitter sends an update whenever the battery's energy level is not zero, and the expected cost in \eqref{bmcost} is not zero, referred to as battery-only with redundancy check (BO-RC) policy.

Fig. \ref{fig:compsd} depicts the average cost $\bar C$ of different policies versus the source's self-transition probability  $p$. As it can be seen for small values of $p$, e.g., $p = 0.4$ or $p = 0.5$, the LC policies and the POMDP-based policy {perform almost the same, and outperform the baseline policies.} This is because, for small values of $p$, the POMDP-based and the LC policies transmit less frequently, while the baseline policies send updates more frequently. Actually, for a high-dynamic source, there is a high chance that the source will change its state in the next slot, and the currently transmitted status update cannot decrease the average cost at the next slot.
{For $p=0.6$, the energy-agnostic LC police's performance becomes the same as the BO-RC policy; this is because, as shown in Fig. \ref{fig:comstrlcp}(a), the energy-agnostic LC policy starts to send updates for large $p$ whenever the battery has energy and the expected distortion cost is not zero, which is the main idea behind the BO-RC policy.
For large values of $p$ (e.g., $p = 0.7,p = 0.8, p = 0.9$), the performance of the BO-RC and energy-agnostic LC policies are very close to the POMDP-based and energy-aware LC policies and outperform the BO policy. This is because when the source state changes slowly, i.e., for large values of $p$, and the channel condition is good, most of the time the expected cost is zero, thus, these policies do not send updates frequently which saves energy for future transmissions. However, the BO policy sends updates independent of the cost value.}


\begin{figure}[!ht]
\centering
\begin{tikzpicture}[scale=1]
\begin{axis}[
    title={},
    xlabel={Source self transition probability, $p$},
    ylabel={Average cost, $\bar C$},
    xmin=0.4, xmax=0.9,
    ymin=0.19, ymax=0.9,
    xtick={0.4,0.5,0.6,0.7,0.8,0.9},
    ytick={0.2,0.3,0.4,0.5,0.60,0.7,0.80,0.9},
    legend pos=north west,
    ymajorgrids=true,
    grid style=dashed,
    legend style={at={(0.05,0.25)},anchor=west},
    legend cell align=left,
]

    \addplot[
    color=Orange,
    mark=o,
    line width=1.3pt, 
    mark size=3pt,
    ]
    coordinates {                       
    (0.4,0.8595)(0.5,0.8086)(0.6,0.7401)(0.7,0.6494)(0.8, 0.5231)(0.9, 0.3290)
    };
    \addplot[
    color=black,
    mark=triangle,
    line width=1.3pt, 
    mark size=3pt,
    ]
    coordinates {   
    (0.4,0.8556)(0.5,0.7905)(0.6,0.7028)(0.7,0.5851)(0.8, 0.4251)(0.9,  0.2255)
    };
    \addplot[
    color=red,
    mark=diamond,
    line width=1.3pt, 
    mark size=3pt,
    ]
    coordinates {
    (0.4,0.6668)(0.5,0.6675)(0.6,0.7010)(0.7, 0.5842)(0.8,  0.4246)(0.9, 0.2272)
    };

    \addplot[
    color=green,
    mark=star,
    line width=1.3pt, 
    mark size=3pt,
    ]
    coordinates {
    (0.4,0.6663)(0.5,0.6663)(0.6,0.6650)(0.7,0.5739)(0.8,0.4224)(0.9,0.2252)
    };
    
\addplot[
    color=blue,
    mark=square,
    line width=1.3pt, 
    mark size=3pt,
    ]
    coordinates {(0.4,6.676e-01)(0.5,6.676e-01)(0.6,6.66e-01)(0.7,5.668e-01)(0.8,4.188e-01)(0.9,2.26e-01)
    };    
    \legend{BO policy, BO-RC policy, Energy-agnostic LC policy, Energy-aware LC policy, POMDP-based policy}



    
\end{axis}
\end{tikzpicture}
  \caption{The average cost $\bar C$ of different policies with respect to $p$ where $N = 3$, $p_s=0.6$, $p_f = 0.6$, $\mu = 0.5$, and $B = 3$.}\label{fig:compsd}
  \end{figure}


Fig. \ref{fig:compenergy} shows the average cost $\bar C$ versus the energy arrival rate $\mu$ for different policies. As it can be seen, for small values of $\mu$, the POMDP-based policy outperforms other policies, while for large values of $\mu$, the performance of the policies except for the BO policy is very close. This is because in cases where the energy arrival rates $\mu$ are small, there are fewer transmission opportunities, and the other policies utilize these opportunities less effectively than the POMDP-based policy. {Moreover, the energy-aware LC policy outperforms the energy-agnostic LC, BO-RC, and BO policies. This is because the energy-aware LC policy uses energy more effectively.}
\begin{figure}[!ht]
\centering
\begin{tikzpicture}[scale=1]
\begin{axis}[
    title={},
    xlabel={Energy arrival rate, $\mu$},
    ylabel={Average cost, $\bar C$},
    xmin=0.2, xmax=0.8,
    ymin=0.5, ymax=0.83,
    xtick={0.2,0.3,0.4,0.5,0.6,0.7,0.8},
    ytick={0.60,0.7,0.80},
    legend pos=north west,
    ymajorgrids=true,
    grid style=dashed,
    legend style={at={(0.25,0.780)},anchor=west},
    legend cell align=left,
]

        \addplot[
    color=orange,
    mark=o,
    line width=1.3pt, 
    mark size=3pt,
    ]
    coordinates {
    (0.2, 7.762006e-01)
(0.3, 7.303979e-01)
(0.4, 6.848223e-01)
(0.5, 6.500065e-01)
(0.6, 6.166515e-01)
(0.7, 5.862279e-01)
(0.8, 5.602844e-01)
    };

    \addplot[
    color=black,
    mark=triangle,
    line width=1.3pt, 
    mark size=3pt,
    ]
    coordinates {
    (0.2, 7.434474e-01)
(0.3, 6.764932e-01)
(0.4, 6.256634e-01)
(0.5, 5.839108e-01)
(0.6, 5.532922e-01)
(0.7, 5.315977e-01)
(0.8, 5.197531e-01)
    };
    \addplot[
    color=red,
    mark=diamound,
    line width=1.3pt, 
    mark size=3pt,
    ]
    coordinates {
    (0.2, 7.434474e-01)
(0.3, 6.764932e-01)
(0.4, 6.23e-01)
(0.5, 5.823e-01)
(0.6, 5.519e-01)
(0.7, 5.315977e-01)
(0.8, 5.197531e-01)
    };
\addplot[
    color=green,
    mark=star,
    line width=1.3pt, 
    mark size=3pt,
    ]
    coordinates {
    (0.2, 6.682802e-01)
(0.3, 6.590e-01)
(0.4, 6.132820e-01)
(0.5, 5.767682e-01)
(0.6, 5.484964e-01)
(0.7, 5.305090e-01)
(0.8, 5.196623e-01)
    };

\addplot[
    color=blue,
    mark=square,
    line width=1.3pt, 
    mark size=3pt,
    ]
    coordinates {
    (0.2, 6.548869e-01)
(0.3, 6.243033e-01)
(0.4, 5.934978e-01)
(0.5, 5.667805e-01)
(0.6, 5.447e-01)
(0.7, 5.296e-01)
(0.8, 5.19e-01)
    };

    \legend{BO policy, BO-RC policy, Energy-agnostic LC policy, Energy-aware LC policy, POMDP-based policy}
    
\end{axis}
\end{tikzpicture}
  \caption{The average cost $\bar C$ of different policies with respect to the energy arrival rate $\mu$, where $p = 0.7$, $N = 3$, $p_f = 0.6$, $p_s=0.6$, and $B = 3$.}\label{fig:compenergy}
  \end{figure}

Fig. \ref{fig:compn} demonstrates the average cost versus the number of source's states $N$ under different policies. From the figure, the average cost increases with $N$, as expected. Moreover, the POMDP-based and energy-aware LC policies outperform other policies, while for large $N$ (e.g., $N=4$ and $N=5$), the POMDP-based policy outperforms the energy-aware LC policy. 
This is because increasing the number of states, in general, leads to a larger estimation error. Consequently, the system needs to use the available energy effectively by considering the state of the system in future time slots. However, except for the POMDP-based and energy-aware LC policies, the other policies are agnostic to the energy arrival rate.

\begin{figure}[!ht]
\centering
\begin{tikzpicture}[scale=1]
\begin{axis}[
    title={},
    xlabel={Number of source states, $N$},
    ylabel={Average cost, $\bar C$},
    xmin=2, xmax=5,
    ymin=0.4, ymax=1.4,
    xtick={2,3,4,5},
    ytick={0.4,0.60,0.80,1,1.2,1.4},
    legend pos=north west,
    ymajorgrids=true,
    grid style=dashed,
    legend style={at={(0.03,0.78)},anchor=west},
    legend cell align=left,
]

    \addplot[
    color=orange,
    mark=o,
    line width=1.3pt, 
    mark size=3pt,
    ]
    coordinates {
    (2,0.4632)(3,0.7750)(4,1.0598)(5,1.3323)
    };

    \addplot[
    color=black,
    mark=triangle,
    line width=1.3pt, 
    mark size=3pt,
    ]
    coordinates {
    (2,0.4598)(3,0.7645)(4,1.0525)(5,1.3253)
    };

\addplot[
    color=red,
    mark=diamond,
    line width=1.3pt, 
    mark size=3pt,
    ]
    coordinates {
    (2,0.4580)(3,0.7577)(4,1.0538)(5,1.3224)
    };

\addplot[
    color=green,
    mark=star,
    line width=1.3pt, 
    mark size=3pt,
    ]
    coordinates {
    (2,0.4462)(3,0.6660)(4,0.9882)(5,1.1994)
    };

\addplot[
    color=blue,
    mark=square,
    line width=1.3pt, 
    mark size=3pt,
    ]
    coordinates {
    (2,0.4449)(3,0.6558)(4,0.8993)(5,1.1304)
    };


    
    \legend{ BO policy, BO-RC policy, Energy-agnostic LC policy, Energy-aware LC policy, POMDP-based policy}
    
\end{axis}
\end{tikzpicture}
  \caption{The average cost $\bar C$ of different policies with respect to numbers of source's state $N$, where $p=0.7$, $N = 3$, $p_s=0.6$, $p_f = 0.2$, $\mu = 0.2$, and $B = 3$.}\label{fig:compn}
  \end{figure}
\indent
In Fig. \ref{fig:comppf}, we examine the impact of the probability of successful reception of ACK $p_f$ for different policies. As can be seen, the average cost decreases by an increase in $p_f$ except for the BO policy. {This is because the number of slots that the sink is observable by the transmitter increases, which in turn, reduces the uncertainty in the system, leading to effective decision-making for those policies that consider the value of the belief. However, the BO policy makes decisions only based on the energy level of the battery, and thus, the reduction in the uncertainty of the system does not change its performance.}
In general, this figure demonstrates 1) the effectiveness of the proposed POMDP-based policy under poor feedback channel conditions and 2) its applicability for low-resource wireless networks, where having high-reliable communication is more costly. 

\begin{figure}[!ht]
\centering
\begin{tikzpicture}[scale=1]
\begin{axis}[
    title={},
    xlabel={Probability of success in feedback channel, $p_f$},
    ylabel={Average cost, $\bar C$},
    xmin=0.05, xmax=1,
    ymin=0.63, ymax=0.72,
    xtick={0.05,0.25,0.5,0.75,1},
    scaled x ticks=false,
    tick label style={/pgf/number format/fixed, /pgf/number format/precision=2}, 
    ytick={0.63,0.65,0.67,0.7,0.72},
    legend pos=north west,
    ymajorgrids=true,
    grid style=dashed,
    legend style={at={(0.26,0.7)},anchor=west},
    legend cell align=left,
]

    \addplot[
    color=orange,
    mark=o,
    line width=1.3pt, 
    mark size=3pt,    
    ]
    coordinates {
    (0.05,0.7138)(0.25,0.7138)(0.5,0.7138)(0.75,0.7138)(1,0.7138)
    };    
    \addplot[
    color=black,
    mark=triangle,
    line width=1.3pt, 
    mark size=3pt,    
    ]
    coordinates {
    (0.05,0.7119)(0.25,0.6997)(0.5,0.6831)(0.75,0.6668)(1,0.6325)
    };

    \addplot[
    color=red,
    mark=diamond,
    line width=1.3pt, 
    mark size=3pt,
    ]
    coordinates {
    (0.05,0.7110)(0.25,0.6982)(0.5,0.6822)(0.75,0.6667)(1,0.6324)
    };

        \addplot[
    color=green,
    mark=star,
    line width=1.3pt, 
    mark size=3pt,
    ]
    coordinates {
    (0.05,0.6658)(0.25,0.6656)(0.5,0.6650)(0.75,0.6598)(1,0.6312)
    };

\addplot[
    color=blue,
    mark=square,
    line width=1.3pt, 
    mark size=3pt,
    ]
    coordinates {
    (0.05,0.66)(0.25,0.66)(0.5,0.655)(0.75,0.645)(1,0.63)
    };

    \legend{ BO policy, BO-RC policy, Energy-agnostic LC policy, Energy-aware LC policy, POMDP-based policy}
    
\end{axis}
\end{tikzpicture}
  \caption{The average cost $\bar C$ of different policies with respect to the probability of successful reception of ACK $p_f$, where $p=0.7$, $N = 3$, $p_s=0.4$, $\mu = 0.5$, and $B = 3$.}\label{fig:comppf}
  \end{figure}
\indent
 {Finally,} Fig. \ref{fig:compB} depicts the average cost with respect to the battery's capacity for different energy arrival rates. Here, without loss of generality, we use the POMDP-based policy. As expected, the average cost decreases by an increase in $B$. This is because a larger capacity of the battery allows the transmitter to store more energy, compensating for periods of inconsistent energy reception and ensuring continuous operation.
The influence of the battery's capacity on performance diminishes with increasing energy arrival rate. This is because, at a higher energy arrival rate, the transmitter receives energy more frequently. 
Therefore, the system can rely more on frequent energy arrival rather than large battery capacities.

\begin{figure}[!ht]
\centering
\begin{tikzpicture}[scale=1]
\begin{axis}[
    title={},
    xlabel={The battery's capacity, $B$},
    ylabel={Average cost, $\bar C$},
    xmin=1, xmax=9,
    ymin=0.5, ymax=0.67,
    xtick={1,3,6,9},
    ytick={0.5,0.55,0.6,0.65,0.7},
    legend pos=north west,
    ymajorgrids=true,
    grid style=dashed,
    legend style={at={(0.65,0.85)},anchor=west},
    legend cell align=left,
]
\addplot[
    color=blue,
    mark=square,
    line width=1.3pt, 
    mark size=3pt,
    ]
    coordinates {
    (1,0.67)(3,0.62)(6,0.592)(9,0.58)
    };

    \addplot[
    color=red,
    mark=diamond,
    line width=1.3pt, 
    mark size=3pt,
    ]
    coordinates {
    (1,0.61)(3,0.56)(6,0.54)(9,0.53)
    };

    \addplot[
    color=black,
    mark=o,
    line width=1.3pt, 
    mark size=3pt,
    ]
    coordinates {
    (1,0.565)(3,0.525)(6,0.515)(9,0.513)
    };

    \legend{  $\mu = 0.3$, $\mu = 0.5$,$\mu = 0.7$}
    
\end{axis}
\end{tikzpicture}
  \caption{The average cost $\bar C$ for different values of the battery's energy arrival $\mu$ with respect to the battery's capacity $B$, where $p=0.7$, $N = 3$, $p_s=0.6$, and ${p_f = 0.7}$.}\label{fig:compB}
  \end{figure}

In Table \ref{TC}, we compare the computational complexity of the proposed policies in the offline phase, i.e., the initial processing time to find a policy, and 2) in the online phase, i.e., running time to find the optimal action at each slot. The computational complexity of the RVIA is at most $\mathcal{O}(|\mathcal{A}_l||\mathcal{L}|^2)$ \cite{zakeri2023minimizing}, where $|\mathcal{A}_l|$ is the size of action space and $|\mathcal{L}|$ is the size of belief-state space. The size of the action space is $2$. If we start from a reset belief, after one transmission attempt without receiving ACK, i.e., $m=1$, there are $N-1$ possible beliefs, and for $m = 2$, there are $N$ possible beliefs. This procedure can start from any reset belief. Thus, the number of beliefs in the belief space is ${N(N-1)N^{m-1} = (N-1)N^{m}}$. Considering the other elements of the belief-state, the size of the belief-state space becomes $2B(N-1)N^{m+2}$. Therefore, the computational complexity for the offline phase of the POMDP-based policy is $\mathcal{O}([B(N-1)N^{m+2}]^2)$. In the online phase, the policy needs to find the action from a state-action table, thus the computational complexity is ${\mathcal{O}(1)}$.
The energy-agnostic LC policy does not have an offline phase. In the online phase, it needs to compare the cost associated with two actions, thus the computational complexity is ${\mathcal{O}(1)}$.
{According to the energy-aware LC policy, one needs to find the best $\gamma$ that provides the minimum average cost, and thus, the complexity for the offline phase is ${\mathcal{O}(1)}$. In the online phase, it is required to compare the cost associated with the two actions, and thus, the computational complexity is ${\mathcal{O}(1)}$.}
 \begin{table}[t]
 \caption{Computational complexity of the proposed policies}
 \centering
 \resizebox{\columnwidth}{!}{ 
 \begin{tabular}{| c | c | c |}
    \hline
    \textbf{Policy} & \textbf{Offline Phase} & \textbf{Online Phase} \\
    \hline
    POMDP-based & \small$\mathcal{O}([B(N-1)N^{m+2}]^2)$ & \small$\mathcal{O}(1)$ \\
    \hline
    Energy-agnostic LC & - & \small$\mathcal{O}(1)$ \\
    \hline
    Energy-aware LC & \small$\mathcal{O}(1)$ & \small$\mathcal{O}(1)$ \\
    \hline
 \end{tabular}
 } 
 \label{TC}
\end{table}

\section{Conclusions}\label{clc}
We considered an energy harvesting status update system monitoring a finite-state Markov source where both forward and feedback channels are error-prone. We addressed the real-time tracking problem by developing policies that minimize the long-term time average of a distortion function subject to an energy constraint imposed by the energy harvesting circuit. We modeled the stochastic problem as a POMDP and then cast it as a belief-MDP problem. By truncating the belief-state space, we proposed a transmission policy using RVIA. Furthermore, we proposed an energy-agnostic LC policy by transforming the belief-MDP problem into a sequence of per-slot problems. Then, we extended the energy-agnostic LC policy to an energy-aware LC policy by adding a regularization term to the objective function of the per-slot problems. We numerically examined the performance and structure of the proposed policies. The results showed that the proposed policies have a switching-type structure. The structures suggest that when the source dynamic is high, the optimal action is to send updates less frequently. In addition, we observed that
{the energy-aware LC policy performs relatively close to the POMDP-based policy. Also, for scenarios where the energy arrival rate is high and the source dynamic is low, the energy-agnostic LC policy performs close to the POMDP-based policy.}


\appendices
\section{Proof of Proposition \ref{pro1}}\label{ap-pro1}
    To show that staying in the most recent successful status update is the optimal ML estimator when $p > q$, we must prove that after $K$ transitions, the probability of self-transition is at least as large as the probability of transitioning to other states. We use the induction approach. We first show that it is valid for one transition. Then, considering that it holds for ${K-1}$ transitions, we prove it holds for ${K}$ transitions.

Let $P$ denote the state transition matrix of the source, given as
    \begin{align}
             P = \begin{bmatrix}
                    p & q&\dots &q\\
                    q & p& & \vdots\\
                    \vdots & &\ddots & q\\
                    q &\dots & q& p
                \end{bmatrix}.
    \end{align}
Note that $P$ is a doubly stochastic matrix, i.e., a square matrix with non-negative entries between $0$ and $1$, where each row and column sums up to $1$. 
The one transition $P^2$ is given as
    \begin{equation}
        \begin{array}{ll}
             P^2 = P\times P = \begin{bmatrix}
                    R_{1,1} & R_{1,2}&\dots & R_{1,2}\\
                    R_{1,2} & R_{1,1}& & \vdots\\
                    \vdots & &\ddots & R_{1,2}\\
                    R_{1,2} &\dots & R_{1,2}& R_{1,1}
                \end{bmatrix},
        \end{array}
    \end{equation}
where ${R_{1,1} = p^2 + (N-1)q^2}$ and ${R_{1,2} = 2pq + (N-2)q^2}$. 
Now, we calculate $R_{1,1}-R_{1,2}$ 
\begin{align}
        R_{1,1}-R_{1,2} &= p^2 + (N-1)q^2 - 2pq - (N-2)q^2\notag\\
        &= p^2 - 2pq + q^2=(p-q)^2.\notag
\end{align}
As $(p-q)^2$ is always larger and equal to zero, we can conclude that $R_{1,1}\geq R_{1,2}$. 


Now, we assume the probability of self-transition is larger and equal to the transition to other states after $K-1$ transition and show it holds for $K$ transitions. The $K$ transitions  $P^{K+1}$ is 
given as
    \begin{equation}
        \begin{array}{ll}
             \hspace{-4mm}P^{K+1} = P\times P^{K} = \begin{bmatrix}
                    R_{K,1} & R_{K,2}&\dots  &R_{K,2}\\
                    R_{K,2} & R_{K,1} &&\vdots\\
                    \vdots & &\ddots & R_{K,2}\\
                    R_{K,2} &\dots & R_{K,2}&R_{K,1}
                \end{bmatrix},
        \end{array}
    \end{equation}
where ${R_{K,1} = pR_{K-1,1} + (N-1)qR_{K-1,2}}$ and ${R_{K,2} = pR_{K,2} +qR_{K-1,1} + (N-2)qR_{K-1,2}}$. 
Now, we calculate $R_{K,1}-R_{K,2}$ to determine which one is larger.
\begin{align}
        &R_{K,1}-R_{K,2} = pR_{K-1,1} + (N-1)qR_{K-1,2} - pR_{K-1,2} \notag\\&\hspace{24mm}-qR_{K-1,1} - (N-2)qR_{K-1,2}\notag\\
        &\hspace{21mm}= (p-q)R_{K-1,1} -(p-q)R_{K-1,2}.\notag
\end{align}
As $p\geq q$ and $R_{K-1,1}\geq R_{K-1,2}$, we can conclude that $R_{K,1}-R_{K,2}\geq 0$, and thus $R_{K,1}\geq R_{K,2}$. That completes the proof.
\section{Derivation of the Conditional Probabilities in~(\ref{case3apx})}\label{case3p1}
The conditional probability ${\mathrm{Pr}(E_1\mid\overline{\text{ACK}})}$ is derived as
    \begin{align}\label{case3e1apx}
             &\hspace{-9mm}\mathrm{Pr}(E_1\mid\overline{\text{ACK}}) = \frac{\mathrm{Pr}(E_1\cap\overline{\text{ACK}})}{\mathrm{Pr}(\overline{\text{ACK}})}\overset{(a)}{=} \frac{\mathrm{Pr}(E_1)}{\mathrm{Pr}(\overline{\text{ACK}})} \notag\\
    &\hspace{12mm}\overset{(b)}{=} \frac{p_s(1-p_f)}{\mathrm{Pr}(\overline{\text{ACK}})}\overset{(c)}{=} \frac{p_s(1-p_f)}{1-p_sp_f},
    \end{align}
where ($a$) comes from the fact that the occurrence of event $E_1$ ensures the occurrence of event $\overline{\text{ACK}}$, ($b$) follows since events $E_1$ happen if the status update is delivered successfully which happens with probability $p_s$, and the sent ACK gets lost which happens with probability $1-p_f$, thus ${\mathrm{Pr}(E_1) = p_s(1-p_f)}$, and ($c$) follows the fact that the transmitter receives ACK if a status update is delivered successfully which happens with probability $p_s$, and the sent ACK is delivered successfully which happens with probability $p_f$, thus the probability of receiving ACK is given as ${\mathrm{Pr}({\text{ACK}}) = p_sp_f}$, and thus we have ${\mathrm{Pr}(\overline{\text{ACK}}) =1-\mathrm{Pr}({\text{ACK}})= 1-p_sp_f}$. 

The conditional probability ${\mathrm{Pr}(E_2\mid\overline{\text{ACK}})}$ is derived as
     \begin{align}\label{case3e2apx}
              \mathrm{Pr}({E_2}\mid\overline{\text{ACK}}) & = \frac{\mathrm{Pr}(E_2\cap\overline{\text{ACK}})}{\mathrm{Pr}(\overline{\text{ACK}})}\overset{(a)}{=} \frac{\mathrm{Pr}(E_2)}{\mathrm{Pr}(\overline{\text{ACK}})}\notag\\&\overset{b}{=} \frac{1-p_s}{1-p_sp_f},\notag
    \end{align}
where ($a$) comes from the fact that the occurrence of event $E_2$ ensures the occurrence of event $\overline{\text{ACK}}$, ($b$) follows since ${\mathrm{Pr}(\overline{\text{ACK}}) = 1-p_sp_f}$, and events $E_2$ happen if the status update is not delivered successfully which happens with probability $1-p_s$, thus $\mathrm{Pr}(E_2) = 1-p_s$. It completes the proof.
\section{Proof of Theorem \ref{theo1}}\label{ap_p_th}

To show the belief-MDP is communicating, it is sufficient to find a randomized policy that induces a recurrent Markov chain, i.e., a policy under which 
{any belief-state ${l'=(X',b',\rho',X_{-1}',f_{-1}')\in \mathcal{L}}$ can be reached from any other state ${l=(X,b,\rho,X_{-1},f_{-1})\in \mathcal{L}}$ with a positive probability} \cite[Prop. 8.3.1]{puterman1994}. 
{We define the following policy: the transmitter sends a status update, i.e., $a(t)=1$, with probability $0.5$ whenever the battery is not empty, i.e.,}
\begin{equation}
\text{Pr}\{a(t) =1 \}=
\begin{cases}
0.5 &  b(t)>0 \\
0 & b(t)=0.
\end{cases}
\end{equation}

{We consider two cases $b'\ge b$ and $b'< b$. For the case where $b'\ge b$, let us consider the following event: for $b'-b$ consecutive slots the transmitter stays idle (i.e., $a=0$) and it receives energy during each slot (i.e., $e=1$), at the next slot, the source state moves to the state regarding the reset belief associated with $\rho'$, the transmitter sends an update (i.e., $a=1$) and receives ACK, and depending on the values of $f_{-1}'$ and $\rho'$ we consider the following three cases: i) if $f_{-1}'=1$ and $\rho'$ is a reset belief, at the next slot, the source's state moves to $X'$, and the transmitter receives one unit of energy, ii) if $f_{-1}'=0$ and $\rho'$ is a reset belief, at the next two slots, first, the source state moves to $X_{-1}'$ and then moves to $X'$, and the transmitter receives one unit of energy, and iii) if $f_{-1}'=0$ and $\rho'$ is not a reset belief, until moving to $\rho'$, at each slot, the transmitter receives one unit of energy, the transmitter sends an update without receiving ACK, and the source state changes in a way that the belief becomes $\rho'$ (see equation \eqref{case3apx}), then, at the next two slots, first, the source state moves to $X_{-1}'$ and then moves to $X'$, and the transmitter receives one unit of energy.
Since the considered event occurs with a positive probability, any belief-state ${l'=(X',b',\rho',X_{-1}',f_{-1}')\in \mathcal{L}}$ can be reached from any other state ${l=(X,b,\rho,X_{-1},f_{-1})\in \mathcal{L}}$ with a positive probability for $b'\ge b$. For the case $b'<b$, first, the battery is discharged to $b'$ by sending updates for $b-b'$ slots without receiving energy, then a similar event with a positive probability can be provided.}

\section{The derivation of \eqref{Eq:expex} }\label{A:expex1}
We start with characterizing ${C(t+1)}$. To this end, first, we introduce the following two auxiliary variables. Let $r(t)\in\{0,1\}$ denote a binary indication of the source state dynamic at slot $t$, where $r(t)=1$ means the source state remains unchanged, and $r(t)=0$ otherwise. Let $y(t)\in\{0,1\}$ denote a binary indicator of the status of the sent packet at slot $t$; $y(t)=1$ means the packet is delivered to the sink, and $y(t)=0$ otherwise. For the sake of brevity, we define ${E=\big(a(t),y(t),r(t+1)\big)}$. The cost for different conditions at the next slot is given as
\begin{align}\label{Eq:expex0}
    &C(t+1) =\notag\\ 
    &\begin{cases}
        d\big(X(t),X(t)\big)&E=\big(1,1,1\big)\\
        \displaystyle\sum_{i=1}^{N}d\big(X(t),i\big)\rho_i(t)&E=(1,0,1)\\
        \displaystyle\sum_{j=1,j\neq X(t)}^{N}d\big(j,X(t)\big)&E=(1,1,0)\\
        \displaystyle\sum_{i=1}^{N}\sum_{j=1,j\neq X(t)}^{N}d(j,i)\rho_i(t)&E=(1,0,0)\\
        \displaystyle\sum_{i=1}^{N}d\big(X(t),i\big)\rho_i(t)&E=(0,0,1)\\
        \displaystyle\sum_{i=1}^{N}\sum_{j=1,j\neq X(t)}^{N}d(j,i)\rho_i(t)&E=(0,0,0).
    \end{cases}
\end{align}

By using \eqref{Eq:expex0}, the conditional expectation of cost at the next slot  ${\mathbb{E}\{C(t+1)\mid l(t)\}}$ is characterized as 
\begin{align}\label{Eq:expex1}
    &\mathbb{E}\{C(t+1)\mid l(t)\} = d\big(X(t),X(t)\big)\mathrm{Pr}\big( E=(1,1,1)   \mid l(t)\big)\notag\\&+\sum_{i=1}^{N}d\big(X(t),i\big)\rho_i(t)\mathrm{Pr}\big( E=(1,0,1)  \mid l(t)\big)\notag\\&+ \sum_{j=1,j\neq X(t)}^{N}d\big(j,X(t)\big)\mathrm{Pr}\big(  E=(1,1,0)  \mid l(t)\big) \notag\\&+ \sum_{i=1}^{N}\sum_{j=1,j\neq X(t)}^{N}\hspace{0mm}d(j,i)\rho_i(t)\mathrm{Pr}\big(  E=(1,0,0)  | l(t)\big) \notag\\&+ \sum_{i=1}^{N}d\big(X(t),i\big)\rho_i(t)\mathrm{Pr}\big(    E=(0,0,1)  \mid l(t)\big)  \notag\\&+  \sum_{i=1}^{N}\sum_{j=1,j\neq X(t)}^{N}d(j,i)\rho_i(t)\mathrm{Pr}\big(  E=(0,0,0)  \mid l(t)\big).
\end{align}
To complete the derivation of \eqref{Eq:expex1}, we need to calculate the conditional probability $\mathrm{Pr}\big(E \mid l(t)\big)$ for different values of $E$. 

The conditional probability $\mathrm{Pr}\big( E\!=\!(1,1,1) \mid l(t)\big)$  is given as
\begin{align}\label{Eq:expex2}
    &\hspace{-1mm}\mathrm{Pr}\big( E\!=\!(1,1,1)\mid l(t)\big)\overset{(a)}{=}\mathrm{Pr}\big(a(t)\!=\!1,y(t)\!=\!1\mid l(t)\big)p,
\end{align}
where ($a$) follows from the fact that the source state dynamic is independent of the transmission decision and packet delivery, and  $\mathrm{Pr}\big(r(t+1)=1\mid l(t)\big)=p$. Next, we calculate the conditional probability ${\mathrm{Pr}\big(a(t)=1,y(t)=1\mid l(t)\big)}$ in \eqref{Eq:expex2}
\begin{align}\label{Eq:expex3}
    &\mathrm{Pr}\big(a(t)=1,y(t)=1\mid l(t)\big) \notag\\&= \mathrm{Pr}\big(y(t)=1\mid a(t)=1, l(t)\big)\mathrm{Pr}\big(a(t)=1\mid l(t)\big) \notag\\&\overset{(a)}{=} p_s\mathrm{Pr}(a(t)=1\mid l(t)) \overset{(b)}{=} p_s\mathbb{E}\{a(t)\mid l(t)\},
\end{align}
where ($a$) follows because ${\mathrm{Pr}\big(y(t)=1\mid a(t)=1, l(t)\big)=p_s}$, and ($b$) comes from the following equality
\begin{align} \label{Eq:expex4}
    &\mathbb{E}\{a(t)\mid l(t)\}=0\mathrm{Pr}(a(t)=0\mid l(t)) + 1\mathrm{Pr}(a(t)=1\mid l(t))\notag\\&= \mathrm{Pr}(a(t)=1\mid l(t)).
\end{align}
By substituting \eqref{Eq:expex3} in \eqref{Eq:expex2}, we have
\begin{align}\label{Eq:expex5}
    &\mathrm{Pr}\big( E=(1,1,1)\mid l(t)\big) =pp_s\mathbb{E}\{a(t)\mid l(t)\}.
\end{align}
The conditional probability $\mathrm{Pr}\big(E \mid l(t)\big)$ for other values of $E$ is derived by following the same steps as for $E=(1,1,1)$.  By substituting the probabilities in \eqref{Eq:expex1} we have
\begin{equation}
    \begin{aligned}
            &\mathbb{E}\{C(t+1)\mid l(t)\} = d\big(X(t),X(t)\big)pp_s\mathbb{E}\{a(t)\mid l(t)\}\notag\\&+\sum_{i=1}^{N}d\big(X(t),i\big)\rho_i(t)p(1-p_s)\mathbb{E}\{a(t)\mid l(t)\}\notag\\&+ \sum_{j=1,j\neq X(t)}^{N}d\big(j,X(t)\big)qp_s\mathbb{E}\{a(t)\mid l(t)\} \notag\\&+ \sum_{i=1}^{N}\sum_{j=1,j\neq X(t)}^{N}\hspace{-0mm}d(j,i)\rho_i(t)q(1\!-\!p_s)\mathbb{E}\{a(t)\mid l(t)\} \notag\\&+ \sum_{i=1}^{N}d\big(X(t),i\big)\rho_i(t)p(1-\mathbb{E}\{a(t)\mid l(t)\})  \notag\\&+  \sum_{i=1}^{N}\sum_{j=1,j\neq X(t)}^{N}d(j,i)\rho_i(t)q(1-\mathbb{E}\{a(t)\mid l(t)\}).
    \end{aligned}
\end{equation}

\bibliographystyle{IEEEtran}
\bibliography{short-conf,short-jour,Main}

\end{document}